\documentclass[aps,pra,reprint,a4paper,10pt,nofootinbib,longbibliography,floatfix]{revtex4-2}
\usepackage[T1]{fontenc}
\usepackage[utf8]{inputenc}
\usepackage{lmodern}
\usepackage{amsmath,amssymb,amsthm,mathtools,bm}
\usepackage{microtype,booktabs,graphicx}
\usepackage{tikz}
\usetikzlibrary{arrows.meta}
\usepackage[dvipsnames]{xcolor}
\usepackage[colorlinks=true,linkcolor=MidnightBlue,citecolor=MidnightBlue,urlcolor=MidnightBlue]{hyperref}
\newtheorem{theorem}{Theorem}
\newtheorem{lemma}{Lemma}
\DeclareMathOperator{\Tr}{Tr}
\DeclareMathOperator{\SEP}{SEP}

\newcommand{\id}{\mathrm{id}}
\newcommand{\ER}{E_R}
\newcommand{\ERi}{E_R^{\infty}}
\newcommand{\cL}{\mathcal L}
\newcommand{\cA}{\mathcal A}
\newcommand{\cN}{\mathcal N}
\newcommand{\cG}{\mathcal G}
\newcommand{\cE}{\mathcal E}
\newcommand{\cF}{\mathcal F}
\newcommand{\cP}{\mathcal P}
\newcommand{\cH}{\mathcal H}
\newcommand{\EE}{\mathbb E}
\newcommand{\fall}[2]{#1^{\underline{#2}}}
\newcommand{\ket}[1]{\lvert #1\rangle}
\newcommand{\bra}[1]{\langle #1\rvert}
\newcommand{\proj}[1]{\ket{#1}\!\bra{#1}}
\newcommand{\norm}[1]{\left\lVert#1\right\rVert}
\newcommand{\dd}{\mathrm d}
\newcommand{\Sym}{\operatorname{Sym}}
\allowdisplaybreaks[2]
\hypersetup{
 pdftitle={Energy-constrained two-way capacity bounds for noisy Gaussian channels},
 pdfauthor={Stefano Pirandola},
 pdfsubject={Thermal loss, noisy amplification, and additive Gaussian noise}
}
\begin{document}
\title{Energy-constrained two-way capacity bounds for noisy Gaussian channels}
\author{Stefano Pirandola}
\affiliation{Department of Computer Science, University of York, York YO10 5GH, United Kingdom}
\begin{abstract}
A gap between achievable rates and converse bounds persists for
energy-constrained bosonic communication in the presence of excess noise.
We derive a unified adaptive weak-converse bound for thermal attenuation,
noisy amplification, and additive Gaussian noise under an unconditional
mean transmitted-photon-number constraint. The bound applies to arbitrary
adaptive protocols with quantum memories and constrains the two-way
quantum, entanglement-distribution, private, and secret-key capacities.
It vanishes throughout the entanglement-breaking region and approaches
the corresponding PLOB bound at infinite energy. At finite energy, it
improves the evaluated Gaussian squashed-entanglement and PLOB bounds in
relevant parameter regimes, while comparison with hashing rates quantifies
the remaining gap to achievability. The result provides a common
energy-dependent benchmark for noisy bosonic Gaussian channels and extends
finite-energy converse methods beyond the quantum-limited setting.
\end{abstract}
\maketitle

\section{Introduction}\label{sec:introduction}
Thermal attenuation, noisy amplification, and additive Gaussian noise are
basic models of imperfect bosonic communication~\cite{HolevoWerner,Weedbrook,BraunsteinVanLoock,HolevoOneMode,CarusoClassification}.
Their energy-constrained two-way-assisted capacities describe the
ultimate rates of quantum communication and secret-key generation with
unlimited authenticated public classical communication at a prescribed
input-energy budget. A finite transmitted-photon budget makes these
rates sensitive to both the channel noise and the energy supplied by the
sender. Noise photons at the receiver do not relax this input constraint.

The PLOB bounds provide energy-independent upper bounds for these three
channel families~\cite{PLOB}. Gaussian squashed-entanglement constructions
give energy-dependent bounds~\cite{DSW,Goodenough}, while coherent and reverse
coherent information yield explicit achievable rates
\cite{DevetakWinter,PGPBL,RCI}. More elaborate key-generation and
entanglement-distillation protocols can improve the lower bounds
\cite{OLCSPB,OPB,MLG,BPinterpolation,BPdistillation,BPthermal}.
For noisy Gaussian channels, these approaches generally leave a gap between
achievable rates and converse bounds.

The distinction between unassisted quantum capacity and capacities with
two-way classical assistance is essential: structural results for
degradable bosonic channels~\cite{WolfCapacity} do not by themselves
determine the adaptive capacities considered here. Entanglement
distillation provides the connection between noisy shared states and
communication rates~\cite{BBPSSW,BDSW}, while channel simulation turns
entanglement bounds into restrictions on adaptive protocols.

Reference~\cite{LossAmplifier} determines the energy-constrained
two-way capacities of the pure-loss and quantum-limited amplifier channels
using \emph{sector teleportation simulation}. This method encodes the
transmitted photon number in an exact simulation resource and accommodates
the amplifier's unbounded output photon number. We extend that framework
to positive excess noise, reusing the simulation, truncation, and
operational testing tools of the same reference. The channel-dependent
step is the evaluation of the mixed noisy resources.

The main additional challenge is that the simulated resource is itself
noisy. Even after fixing the total input and output photon numbers, the
state retains uncertainty over the different attenuation and amplification
processes that may have occurred. We therefore evaluate its entanglement
directly and show that the resulting quantity has the structure needed for
the converse. In the large-sector limit, this leads to a common
energy-dependent function that applies to attenuation, amplification, and
unit-gain additive noise.

The resulting converse applies to arbitrary correlated adaptive inputs.
A photon-number decomposition introduces only a sublinear correction, which
vanishes asymptotically, while the large-sector entanglement density allows
the bound to be expressed directly in terms of the actual unconditional
transmitted energies. These ingredients produce energy-dependent converse
bounds that improve the previously available bounds and, for the pure-loss
and quantum-limited amplifier channels, meet the corresponding achievable
rates.
More broadly, the results show that retaining the physical input-energy
budget within channel simulation can turn otherwise energy-insensitive
entanglement bounds into sharp finite-energy converses for adaptive bosonic
communication.

Section~\ref{sec:results} presents the main results and comparisons, and
Sec.~\ref{sec:conclusion} gives the conclusions.
Appendix~\ref{app:structure} outlines the structure of the proof and
summarizes its dependencies in a flow diagram.
Appendix~\ref{app:tools} organizes the inherited tools around the
operational target, sector simulation, and photon-number accounting;
Appendix~\ref{app:resource} evaluates the corresponding resource states,
Appendix~\ref{app:converse} proves the common adaptive weak converse, and
Appendix~\ref{app:benchmarks} treats limits and benchmarks.

\section{Main results}\label{sec:results}
\subsection{Channels and operational assumptions}\label{sec:model}
We use quadratures for which the vacuum covariance is $I_2:=\mathrm{diag}(1,1)$. A centered,
phase-insensitive Gaussian channel $\cG_{\tau,v}$ is specified by
\begin{equation}
 \bar r\longmapsto\sqrt\tau\,\bar r,~~
 V\longmapsto\tau V+(2v+1-\tau)I_2.
 \label{eq:channel}
\end{equation}
Here $\tau>0$ is the intensity gain and $v$ is the output mean photon
number for vacuum input. Complete positivity requires
$v\ge\max\{0,\tau-1\}$. We consider the strictly noisy region
\begin{equation}
 v>\max\{0,\tau-1\},~~
 \Tr[\hat n\,\cG_{\tau,v}(\rho)]=\tau\Tr(\hat n\rho)+v.
 \label{eq:mean-output}
\end{equation}
Table~\ref{tab:channels} identifies the three families. For thermal loss
and noisy amplification, $\nu>0$ is the environmental thermal photon
number. For additive noise our convention is
\begin{equation}
 \cN_\xi(\rho)=\int_{\mathbb C}\frac{\dd^2\alpha}{\pi\xi}
 e^{-|\alpha|^2/\xi}D(\alpha)\rho D(\alpha)^\dagger,
 ~~\xi>0,
 \label{eq:additive-definition}
\end{equation}
so $\xi$ is the number of photons added to a vacuum input.
These channels are entanglement breaking precisely when $v\ge\tau$
\cite{HolevoWerner,PLOB,MLG}.

\begin{table}[tb]
\vspace{-0.2cm}
\caption{Channel parameters in Eq.~\eqref{eq:channel}. The thermal
environmental photon number is $\nu>0$ in the first two rows.}
\label{tab:channels}
\begin{ruledtabular}
\begin{tabular}{lcc}
Channel & $\tau$ & $v$\\
\hline
Thermal loss $\cL_{\eta,\nu}$ & $0<\eta<1$ & $(1-\eta)\nu$\\
Noisy amplifier $\cA_{\kappa,\nu}$ & $\kappa>1$ & $(\kappa-1)(\nu+1)$\\
Additive noise $\cN_\xi$ & $1$ & $\xi>0$
\end{tabular}
\end{ruledtabular}
\end{table}


Alice and Bob start from local preparations and public randomness,
without shared entanglement or secret key. Write $\cG_{A\to B}$ for
the memoryless channel, with input mode $A$ and output mode $B$,
and let $X$ and $Y$ denote Alice's and Bob's local memories,
respectively. Immediately before channel use $i$, their unconditional
joint state is $\rho^{(i)}_{X_iA_iY_i}$: Alice holds $X_iA_i$ and
Bob holds $Y_i$. The channel produces
\begin{equation}
 \omega^{(i)}_{X_iB_iY_i}
 =(\mathrm{id}_{X_i}\otimes\cG_{A_i\to B_i}
   \otimes\mathrm{id}_{Y_i})
   (\rho^{(i)}_{X_iA_iY_i}),
\end{equation}
after which Alice holds $X_i$ and Bob holds $B_iY_i$.
The memories include all retained local quantum registers and
classical histories. The marginal $\rho^{(i)}_{A_i}$ is the actual
transmitted state, and $\hat n_{A_i}$ denotes the input
photon-number operator. We omit the index $i$ when discussing
a generic channel use.

Between channel uses, the parties perform adaptive local operations
and authenticated two-way public classical communication (LOCC).
Each finite protocol has finitely many local rounds and finite or
countable outcomes; its round bound may grow with the number of
channel uses. All Hilbert spaces are separable, states are normal, and local
instruments are normal. Local memories and apparatus have no dimension,
energy, or entropy constraint. Eve holds the channel environments,
their purifications, and every public message. These are the
operational assumptions defining the energy-constrained capacities
studied here, as in Ref.~\cite[Appendix A 2]{LossAmplifier}. After the $n$th channel use and the final LOCC, we denote by
$\zeta_n$ the unconditional joint state of all registers retained
by Alice and Bob, including any retained local auxiliary systems,
with Eve's systems traced out. For key protocols, the entanglement analysis uses $\zeta_n$
for the final state of the input-preserving retained quantum
implementation, including its local shields, before the key
registers are measured; see Appendix~\ref{app:tests}.
The corresponding measured state
$\kappa_{K_AK_BE_{\rm all}}$, introduced below, is used to
define the operational key error.

For an $n$-use protocol, constrain only the actual unconditional
transmitted means:
\begin{equation}
 x_i=\Tr(\hat n_{A_i}\rho^{(i)}_{A_i}),~~
 \sum_{i=1}^n x_i\le nN<\infty.
 \label{eq:energy}
\end{equation}
For a fixed
channel $\cG$ and budget $N$, the energy-constrained capacities $Q_2=D_2\le P_2=K_2$ concern quantum transmission, entanglement distribution,
private classical transmission, and secret-key agreement, respectively.
Each is the supremum of achievable rates $\liminf_n n^{-1}\log_2d_n$
over protocol sequences obeying Eq.~\eqref{eq:energy} with vanishing
error and finite target size $d_n$ at every blocklength. 


At blocklength $n$, write $d=d_n$. For entanglement distribution,
the target is the maximally entangled state
\begin{equation}
 \Phi_d=\proj{\Phi_d},~~
 \ket{\Phi_d}=\frac{1}{\sqrt d}\sum_{a=0}^{d-1}
 \ket{a}_{Q_A}\ket{a}_{Q_B},
\end{equation}
where $Q_A$ and $Q_B$ are $d$-dimensional registers held by Alice
and Bob. Let $\sigma_{Q_AQ_B}$ be their actual output state,
obtained from $\zeta_n$ by tracing out the other retained systems.
The quantum error is the entanglement infidelity
\begin{equation}
 \varepsilon=1-\Tr(\Phi_d\,\sigma_{Q_AQ_B}).
\end{equation}
Quantum transmission uses the same criterion: Alice prepares
$\Phi_d$ locally and transmits one half through the protocol,
while retaining $Q_A$ as a reference.

For secret-key agreement, the target consists of identical,
uniform classical keys that are independent of Eve. Its key
marginal is
\begin{equation}
 \tau_d^{\rm key}
 =\frac1d\sum_{a=0}^{d-1}\proj{aa}_{K_AK_B},
\end{equation}
where $K_A$ and $K_B$ are Alice's and Bob's classical key registers.
Let $\kappa_{K_AK_BE_{\rm all}}$ be the actual final state after
key readout, with $E_{\rm all}$ containing Eve's quantum side
information and the complete public transcript. The ideal
comparison state is
$\tau_d^{\rm key}\otimes\kappa_{E_{\rm all}}$, which retains
Eve's actual marginal. The key error is
\begin{equation}
 e=\tfrac12\norm{\kappa_{K_AK_BE_{\rm all}}
       -\tau_d^{\rm key}\otimes\kappa_{E_{\rm all}}}_1,
 \label{eq:key-error}
\end{equation}
in trace norm $\norm{Z}_1=\Tr\sqrt{Z^\dagger Z}$. This criterion jointly expresses correctness and secrecy
\cite{PortmannRenner}. Private classical transmission uses the
same target and error criterion, with $K_A$ replaced by a
uniformly chosen message register $M$ and $K_B$ by Bob's decoded
message register $\widehat M$.

\subsection{A common energy-constrained weak converse}\label{sec:unified}
Let us define the entropic functions
\begin{align}
 &g(x)=(x+1)\log_2(x+1)-x\log_2x,\nonumber \\
 &\varphi(x)=x\log_2x,
 \label{eq:entropy-functions}
\end{align}
with $g(0)=\varphi(0)=0$. Then, for $N\ge0$, we also set
\begin{equation}
 \begin{aligned}
 y&=\tau N+v,~~ L=1+N+y,\\
 t&=\frac{L-\sqrt{L^2-4\tau N(N+1)}}2,\\
 u&=N-t,~~ w=y-t.
 \end{aligned}
 \label{eq:rateparams}
\end{equation}
The parameters $u,w$ are nonnegative. In the non-entanglement-breaking
region $v<\tau$, we define
\begin{equation}
 \boxed{\begin{aligned}
 \cF_{\tau,v}(N)={}&\varphi(t)+\varphi(L-t)-\varphi(N)\\
 &-\varphi(y)-g(u)-g(w).
 \end{aligned}}
 \label{eq:F}
\end{equation}
For $v\ge\tau$, set $\cF_{\tau,v}(N)=0$.

\begin{theorem}[Energy-constrained weak converse]\label{thm:gaussian}
For every channel in Table~\ref{tab:channels}, under the operational
assumptions and energy constraint of Eq.~\eqref{eq:energy}, we can write
\begin{equation}
 \boxed{\begin{aligned}
 Q_2(\cG_{\tau,v},N)&=D_2(\cG_{\tau,v},N)\\
 &\le P_2(\cG_{\tau,v},N)=K_2(\cG_{\tau,v},N)\\
 &\le\cF_{\tau,v}(N).
 \end{aligned}}
 \label{eq:capacities}
\end{equation}
\end{theorem}
\noindent Note that the function $\cF_{\tau,v}$ is nonnegative, nondecreasing, and concave
on $[0,\infty)$, with $\cF_{\tau,v}(0)=0$. It is at most $g(N)$ and
approaches the corresponding PLOB bound for large $N$.
These properties follow from Lemmas~\ref{lem:density}
and~\ref{lem:concavity} and Appendices~\ref{app:limits}
and~\ref{app:highenergy}.

The statement applies to arbitrary adaptive inputs with quantum
memories. It neither restricts the protocol to Gaussian operations nor
asserts equality between the quantum and key capacities. The
energy-constrained capacities of the two quantum-limited boundary
channels are determined in Ref.~\cite{LossAmplifier}; the present
theorem specifically concerns positive excess noise.

\subsection{Specification of the bound to the three channels}\label{sec:three}
We now study the specification of the common bound for the three different channels of Table~\ref{tab:channels}.

For \emph{thermal loss}, we have
\begin{equation}
 \cF^{\rm loss}_{\eta,\nu}(N)=\cF_{\eta,(1-\eta)\nu}(N).
 \label{eq:loss-bound}
\end{equation}
For $0<\nu<\eta/(1-\eta)$, its infinite-energy PLOB limit is given by~\cite{PLOB}
\begin{equation}
 \Phi^{\rm loss}_{\eta,\nu}
 =-\log_2(1-\eta)-\nu\log_2\eta-g(\nu).
 \label{eq:loss-plob}
\end{equation}
Also note that, for $\nu\ge\eta/(1-\eta)$ all four capacities vanish.

For a \emph{noisy amplifier}, the added photon number contains both the
vacuum contribution and the thermal contribution. So we write
\begin{equation}
 \cF^{\rm amp}_{\kappa,\nu}(N)
 =\cF_{\kappa,(\kappa-1)(\nu+1)}(N),~~\nu>0.
 \label{eq:amp-bound}
\end{equation}
For $0<\nu<1/(\kappa-1)$, the PLOB limit is~\cite{PLOB}
\begin{equation}
 \Phi^{\rm amp}_{\kappa,\nu}
 =(\nu+1)\log_2\kappa-\log_2(\kappa-1)-g(\nu).
 \label{eq:amp-plob}
\end{equation}
Again, for $\nu\ge1/(\kappa-1)$, all four capacities vanish.

Finally, for \emph{additive noise}, set $\tau=1$ and $v=\xi$. For $0<\xi<1$,
Eq.~\eqref{eq:F} can be re-written as follows:
\begin{equation}
 \begin{aligned}
 \cF^{\rm add}_\xi(N)={}&\varphi(N-u_\xi)
       +\varphi(N+u_\xi+\xi+1)\\
 &-\varphi(N)-\varphi(N+\xi)\\
 &-g(u_\xi)-g(u_\xi+\xi),
 \end{aligned}
 \label{eq:add-bound}
\end{equation}
where we have introduced
\begin{equation}
u_\xi=\frac{\sqrt{(1+\xi)^2+4\xi N}-(1+\xi)}2.
 \label{eq:add-u}
\end{equation}
Its infinite-energy PLOB limit is~\cite{PLOB}
\begin{equation}
 \Phi^{\rm add}_\xi=-\log_2\xi+(\xi-1)\log_2e,
 \label{eq:add-plob}
\end{equation}
and all four capacities vanish for $\xi\ge1$. 

In each non-entanglement-breaking
case, monotonicity gives $\cF(N)\le\Phi$ at finite energy. The three
large-energy expansions are proved in Appendix~\ref{app:highenergy}.

\subsection{Proof sketch}\label{sec:finite}
We prove Theorem~\ref{thm:gaussian} by first establishing
finite-error converse bounds. The argument combines an upper
bound on the entanglement generated by an adaptive protocol
with a lower bound on the entanglement required by its
approximate target state. The capacity bounds then follow
by taking the vanishing-error limit.

Let $\ER$ denote relative entropy of entanglement and $\ERi$ its
regularization, as explicitly defined in Appendix~\ref{app:ree}.
The central tool is sector teleportation simulation, recalled
from Ref.~\cite[Appendices B 1 and E 1]{LossAmplifier} and applied
to the noisy channels in Appendix~\ref{app:simulation}.
For a block of $s$ input modes, consider the finite-dimensional
sector $\cH_{s,m}$ containing exactly $m$ photons in total.
Let $\Phi_{s,m}$ be a maximally entangled state between a reference
$R$ and a copy of this input sector. The associated resource is
\begin{equation}
 \mathcal J_{s,m}
 =(\mathrm{id}_R\otimes\cG^{\otimes s})(\Phi_{s,m}).
\end{equation}
Alice holds the reference share and Bob holds the channel-output
share.


Passive mode rotations act irreducibly within $\cH_{s,m}$, and
the channel is covariant under these rotations. These properties
provide a finite-outcome LOCC map $\mathcal T_{s,m}$ satisfying
\begin{equation}
 \begin{aligned}
 &\mathcal T_{s,m}
   (\rho_{XA^sY}\otimes\mathcal J_{s,m})\\
 &~~=
 (\mathrm{id}_X\otimes\cG^{\otimes s}\otimes\mathrm{id}_Y)
 (\rho_{XA^sY}),
 \end{aligned}
 \label{eq:simulation-main}
\end{equation}
for every input $\rho_{XA^sY}$ whose optical subsystem $A^s$
is supported on $\cH_{s,m}$. The same map works for all such
inputs, including arbitrary correlations with the memories
$X,Y$; its identity action on these memories is implicit.
Alice performs a suitable Bell measurement on $A^s$ and her
resource share, communicates the outcome, and Bob applies
the corresponding passive unitary. Thus one copy of
$\mathcal J_{s,m}$ and LOCC reproduce the channel output exactly.

Since the resource is initially independent of the input,
LOCC monotonicity and subadditivity bound the entanglement
increase by $f_s(m):=\ER(\mathcal J_{s,m})$. In this way, the
transmitted photon number enters the converse through the
entanglement of the simulation resource.



The resource calculation in Appendix~\ref{app:resource},
specifically Lemma~\ref{lem:density}, establishes
\begin{equation}
 \lim_{s\to\infty}\frac{f_s(m_s)}{s}
 =\cF_{\tau,v}(N)
 ~~\text{whenever}~~
 \frac{m_s}{s}\longrightarrow N.
\end{equation}
Thus $\cF_{\tau,v}$ in Eq.~\eqref{eq:F} is the asymptotic
resource entanglement per mode. Passing from this sector
statement to a general input requires accounting for coherence
between different total photon numbers.

Consider an arbitrary correlated input $\rho_{XAY}$ with
$\ER(XA:Y)_\rho<\infty$ and transmitted mean
$x=\Tr(\hat n_A\rho_A)<\infty$, and let $\omega_{XBY}$ be its
channel output. Apply the block argument of
Appendix~\ref{app:block} to $s$ independent replicas of this fixed
state. Sector simulation bounds the contribution within each
photon-number sector by its resource cost. The additional
contribution from coherence between sectors is bounded by the
entropy of Bob's total output photon count, which is at most
$g(s(\tau x+v))$. After division by $s$, this overhead vanishes.
Meanwhile, the averaging argument in Appendix~\ref{app:replicas}
shows that the average sector cost per mode converges to
$\cF_{\tau,v}(x)$. Taking the regularization limit therefore
gives Lemma~\ref{lem:gain}:
\begin{equation}
 \ERi(X:BY)_\omega
 \le\ERi(XA:Y)_\rho+\cF_{\tau,v}(x).
 \label{eq:gain-main}
\end{equation}
The replicas are an auxiliary proof construction for a fixed
state; this bound applies to the actual correlated input at
each use of an adaptive protocol.

To combine uses with different input energies, we also require
concavity and monotonicity of $\cF_{\tau,v}$. These properties
are established in Lemma~\ref{lem:concavity} of
Appendix~\ref{app:concavity}. The same lemma proves the additional
identity $\cF_{\tau,v}(N)=\ERi(\omega_N)$, where $\omega_N$ is
obtained by transmitting one arm of a two-mode squeezed vacuum
with mean photon number $N$ in each arm.

Starting from a separable state, iterating
Eq.~\eqref{eq:gain-main} and using LOCC monotonicity between
channel calls bounds the entanglement of the final retained
state $\zeta_n$. Concavity, monotonicity, and the energy
constraint then give
\begin{equation}
 \ERi(\zeta_n)\le\sum_{i=1}^n\cF_{\tau,v}(x_i)
                   \le n\cF_{\tau,v}(N).
 \label{eq:adaptive}
\end{equation}
The adaptive argument is detailed in Appendix~\ref{app:adaptive}.

Equation~\eqref{eq:adaptive} limits the entanglement available
at the protocol output. To obtain an operational bound, we
compare it with the entanglement required by an approximate
target of size $d$. The output-test bound recalled in
Appendix~\ref{app:tests} gives
\begin{equation}
 \ERi(\zeta_n)\ge(1-\delta)\log_2d-h_2(\delta),
\end{equation}
where $\delta$ bounds the failure probability of the corresponding
entanglement or privacy test
\cite[Appendices C 2 and D 2]{LossAmplifier}.
For a quantum target, $\delta=\varepsilon$.
For a key protocol, the input-preserving reduction of
Ref.~\cite[Appendix D 3]{LossAmplifier} supplies a retained quantum
implementation, including local shields, whose privacy-test
failure is at most $\delta=2e-e^2$. This reduction preserves
every transmitted marginal and hence every $x_i$.
Combining these output requirements with
Eq.~\eqref{eq:adaptive} yields the following lemma.

\begin{lemma}[Finite-error converse bounds]\label{lem:finite}
An $n$-use protocol with target size $d$ satisfies
\begin{align}
 (1-\varepsilon)\log_2d
 &\le\sum_{i=1}^n\cF_{\tau,v}(x_i)+h_2(\varepsilon),
 \label{eq:quantum-error}\\
 (1-e)^2\log_2d
 &\le\sum_{i=1}^n\cF_{\tau,v}(x_i)+h_2(2e-e^2),
 \label{eq:key-bound}
\end{align}
for quantum entanglement infidelity $\varepsilon$ and operational
key error $e$, respectively. Here $h_2$ is binary entropy, and
each resource sum is at most $n\cF_{\tau,v}(N)$.
\end{lemma}
In Appendix~\ref{app:operational} we provide details on this conversion from
entanglement accounting to operational rates. Finally, dividing by $n$
and taking vanishing error gives the weak-converse capacity
bounds of Theorem~\ref{thm:gaussian}.

\begin{figure*}[t]
\centering
\includegraphics[width=0.86\textwidth]{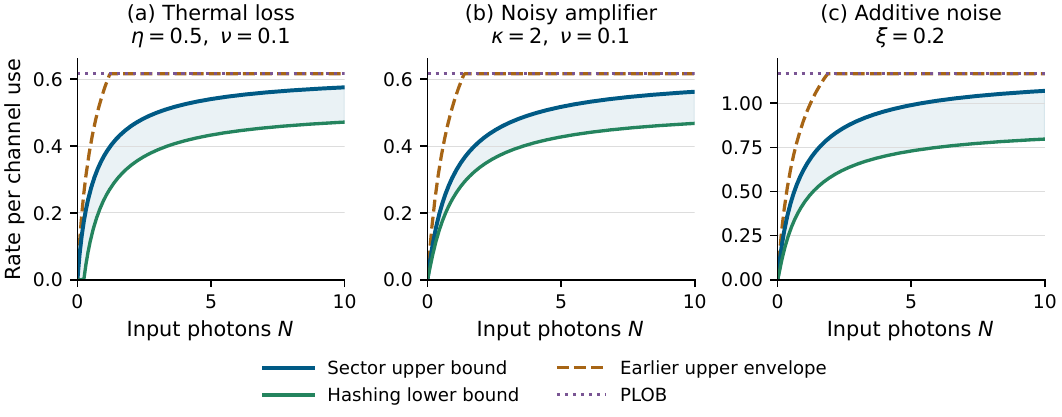}
\vspace{-0.3cm}
\caption{Bounds on energy-constrained capacities for (a) thermal loss with $\eta=1/2$,
$\nu=0.1$; (b) noisy amplification with $\kappa=2$, $\nu=0.1$; and
(c) additive noise with $\xi=0.2$. The blue solid curve is our sector
upper bound, the green solid curve is the hashing lower bound, the
orange dashed curve is $U_{\rm prev}$ of Eq.~\eqref{eq:previous}, and
the purple dotted line is the energy-independent PLOB bound $\Phi$. Shading
marks the gap between the displayed achievable rate and sector bound.
The envelope and PLOB curves coincide where PLOB is the smallest of
the three earlier bounds.}
\label{fig:bounds}
\end{figure*}

\subsection{Achievable benchmarks and comparisons}\label{sec:comparison}
The Gaussian output $\omega_N$ has entropies $S(R)=g(N)$,
$S(B)=g(y)$, and $S(RB)=g(u)+g(w)$. Hashing in the better communication
direction achieves~\cite{DevetakWinter,PGPBL,RCI}
\begin{equation}
 R_{\rm hash}(N)=
 \big[\max\{g(N),g(y)\}-g(u)-g(w)\big]_+,
 \label{eq:hashing}
\end{equation}
where $[z]_+=\max\{z,0\}$. The larger marginal is always Bob's for the
noisy amplifier and additive-noise channel. Consequently,
\begin{equation}
 R_{\rm hash}(N)\le Q_2=D_2\le P_2=K_2\le\cF_{\tau,v}(N).
 \label{eq:bracket}
\end{equation}
This is an explicit achievable benchmark, without an assertion that
hashing is optimal among available distillation protocols. We also evaluate the balanced Gaussian squashing constructions of
Refs.~\cite{DSW,Goodenough}, denoted $U_{\rm DSW}$ and $U_{\rm GEW}$, and
compare with
\begin{equation}
 U_{\rm prev}(N)=\min\{U_{\rm DSW}(N),U_{\rm GEW}(N),\Phi\}.
 \label{eq:previous}
\end{equation}

Figure~\ref{fig:bounds} displays all three channel families (Appendix~\ref{app:numerics} specifies the covariance-matrix evaluations.) From the figure, we can appreciate the large improvement achieved by our tighter upper bound.



\section{Conclusions}\label{sec:conclusion}
We have obtained a common energy-constrained adaptive weak converse for
thermal-loss, noisy amplifier, and additive-noise Gaussian channels.
The input-energy dependence is retained through exact photon-sector
simulation resources. Their optimal separable comparisons and
asymptotic entanglement density produce a concave energy function that
bounds the entanglement increase at each use, including for arbitrary
correlated memories. The same function vanishes in the
entanglement-breaking region and recovers the corresponding PLOB
bound at infinite energy.

The finite-energy examples improve the evaluated upper-bound
benchmarks and leave an explicit gap to hashing. Closing this gap
requires further achievability results or a tighter converse.
The present claims concern vanishing-error energy-constrained
capacities. Public vacuum dilution under an unconditional mean-energy
constraint, analyzed in Ref.~\cite[Appendix F 1]{LossAmplifier},
prevents inferring an
energy-dependent strong converse merely from the weak-converse bound.
The exact strong-converse thresholds of the noisy channels are not
determined here. Finite-blocklength refinements and tighter photon-number
constraints provide further directions for the sector method.

These results supply energy-dependent benchmarks for the design and
assessment of bosonic entanglement-distribution and key-generation
protocols in the presence of excess noise. Comparing achievable rates
with the bounds at the same input budget identifies parameter regimes
where further improvements could have the greatest value. At finite
blocklength, Lemma~\ref{lem:finite} also limits the entanglement or key
size compatible with a prescribed photon budget and error tolerance.
More broadly, the sector-simulation framework connects physical input
constraints to resource entanglement, suggesting a route to analogous
converse bounds for other covariant channels.

\par\medskip
\noindent\textbf{AI statement.}
AI's ChatGPT (GPT-6 Astra model) was used during the
preparation of this manuscript to generate candidate mathematical
derivations and proofs, assist with their development, and improve the presentation of the manuscript. The tool was directed through
author-specified prompts concerning the mathematical assumptions,
derivations, and statements under consideration. All AI-assisted
mathematical results, derivations, and proofs were independently
examined, verified, and revised as necessary by the author, who assumes responsibility for the accuracy and scientific content of the work.

\onecolumngrid
\appendix

\section{Structure of the proof}\label{app:structure}
The proof separates the physical channel simulation, the entanglement
calculation for its resource, and the conversion to communication rates.
Figure~\ref{fig:proof-flow} shows how these ingredients fit together.
The simulation and output tests are inherited from
Ref.~\cite{LossAmplifier}; the noisy-resource evaluation and its use in
the energy-dependent converse are developed below. The lemma numbers
follow their order of appearance in the manuscript, while the arrows
show the order in which their conclusions are used.

\begin{figure}[!ht]
\centering
\begin{tikzpicture}[
  box/.style={draw=MidnightBlue!75,fill=MidnightBlue!4,rounded corners=2pt,
    align=center,font=\small,inner sep=6pt,minimum height=1.15cm},
  wide/.style={box,text width=7.7cm},
  branch/.style={box,text width=5.5cm},
  link/.style={-{Stealth[length=2.2mm]},thick,draw=MidnightBlue!80},
  testbox/.style={box,draw=ForestGreen!75!black,fill=ForestGreen!5,
    text width=2.9cm,font=\footnotesize}
]
\node[wide] (sim) at (0,0)
  {\textbf{Exact sector teleportation simulation}\\
   Fixed input number $m$ in $s$ modes; resource $\mathcal J_{s,m}$\\
   Appendix~\ref{app:simulation}};
\node[branch] (spectrum) at (-3.4,-1.9)
  {\textbf{Evaluate the noisy resource}\\
   Optimal comparison and product additivity\\
   Lemma~\ref{lem:optimal}; Appendix~\ref{app:resource}};
\node[branch] (density) at (-3.4,-3.8)
  {\textbf{Take the sector-density limit}\\
   $f_s(m_s)/s\to\cF_{\tau,v}(N)$\\
   Lemma~\ref{lem:density}};
\node[branch] (block) at (3.4,-3.8)
  {\textbf{Treat arbitrary correlated inputs}\\
   Sector cost plus count entropy\\
   Cutoffs removed\\
   Lemma~\ref{lem:block}; Appendix~\ref{app:block}};
\node[wide] (gain) at (0,-5.7)
  {\textbf{Replica limit: one-use entanglement gain}\\
   Regularize a fixed state; the count overhead vanishes\\
   Lemma~\ref{lem:gain}; Appendix~\ref{app:replicas}};
\node[wide] (concavity) at (0,-7.7)
  {\textbf{Gaussian identity and concavity}\\
   $\ERi(\omega_N)=\cF_{\tau,v}(N)$; classical energy flags\\
   Lemma~\ref{lem:concavity}; Appendix~\ref{app:concavity}};
\node[wide] (adaptive) at (0,-9.7)
  {\textbf{Adaptive entanglement accounting}\\
   $\ERi(\zeta_n)\le\sum_i\cF_{\tau,v}(x_i)\le n\cF_{\tau,v}(N)$\\
   Appendix~\ref{app:adaptive}};
\node[testbox] (tests) at (6.1,-11.7)
  {\textbf{Output tests}\\
   Entanglement tests\\
   Privacy tests\\
   Appendix~\ref{app:tests}};
\node[wide] (finite) at (0,-11.7)
  {\textbf{Finite-error converse bounds}\\
   Available entanglement versus target requirement\\
   Lemma~\ref{lem:finite}; Appendix~\ref{app:operational}};
\node[wide,fill=ForestGreen!6,draw=ForestGreen!75!black] (theorem)
  at (0,-13.7)
  {\textbf{Energy-constrained weak converse}\\
   Divide by $n$ and let the operational error vanish\\
   Theorem~\ref{thm:gaussian}};
\draw[link] (sim.south west) -- (spectrum.north);
\draw[link] (sim.south east) -- (block.north);
\draw[link] (spectrum) -- (density);
\draw[link] (density.south) -- (gain.north west);
\draw[link] (block.south) -- (gain.north east);
\draw[link] (gain) -- (concavity);
\draw[link,dashed] (spectrum.west) -- (-7,-1.9)
  -- (-7,-7.7) -- (concavity.west);
\draw[link] (concavity) -- (adaptive);
\draw[link] (adaptive) -- (finite);
\draw[link,draw=ForestGreen!75!black] (tests.west) -- (finite.east);
\draw[link] (finite) -- (theorem);
\end{tikzpicture}
\caption{Dependencies of the weak-converse proof. The dashed arrow
records the additional use of product additivity in the Gaussian-state
density and concavity argument. The sector and replica limits are
auxiliary mathematical limits; the adaptive protocol itself may have
correlated inputs at every use. Appendix~\ref{app:benchmarks} then
evaluates channel limits and the compared achievable and upper bounds.}
\label{fig:proof-flow}
\end{figure}
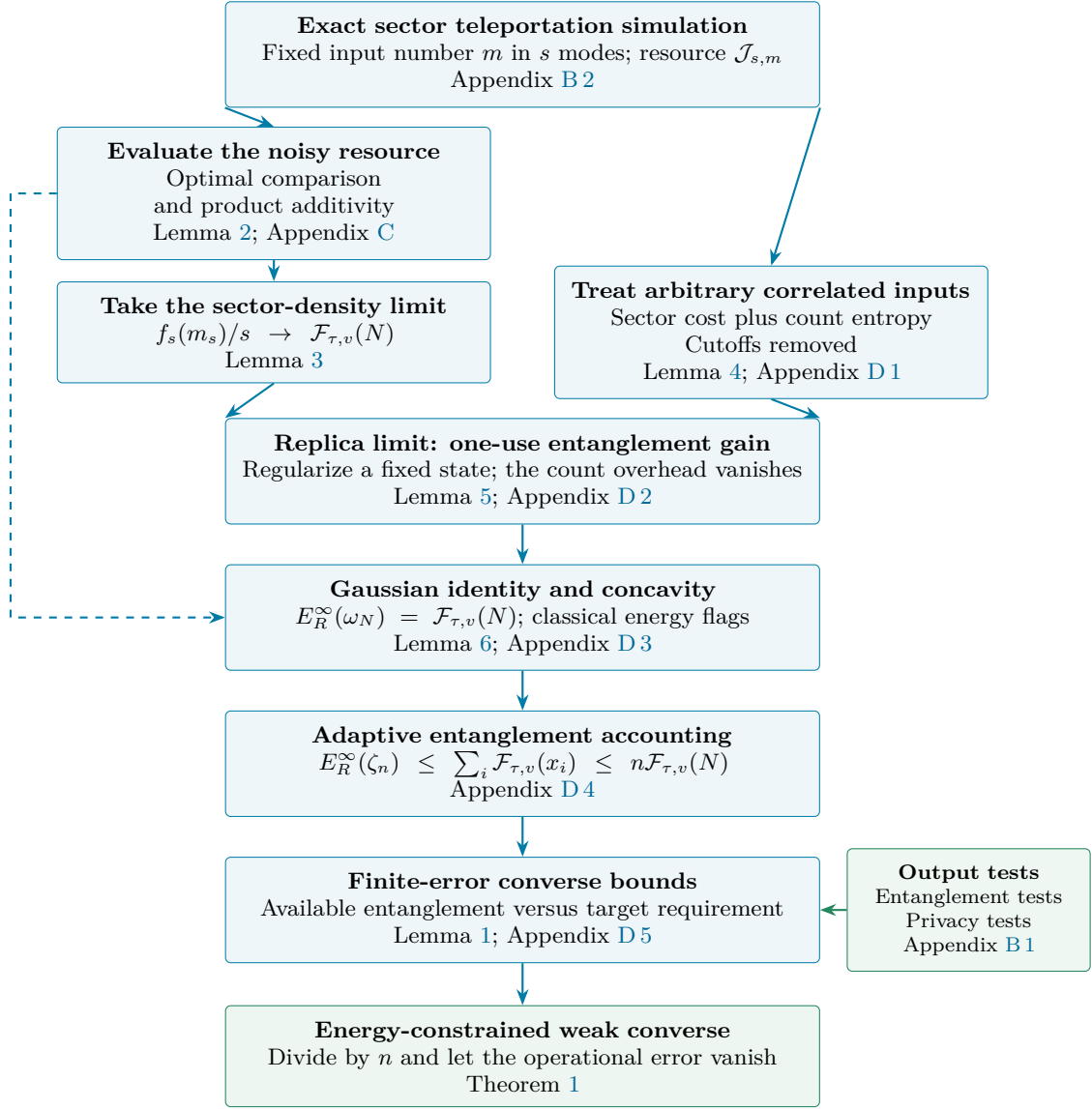
\clearpage

\section{Reduction to the noisy-resource calculation}\label{app:tools}
The converse has two distinct ingredients. An operational test tells us
how much entanglement is required at the end of a successful protocol.
A simulation argument bounds how much entanglement each transmission
can add. Photon-number resolution connects the simulation to the
transmitted energy. This appendix recalls these ingredients in the form
needed here, using the protocol and error conventions of
Sec.~\ref{sec:model}.

The simulation, truncation, and target-testing results are taken from the
common pure-loss and quantum-limited amplifier analysis in
Ref.~\cite{LossAmplifier}. We recall their statements and explain their
role, while referring to that paper for their general proofs. The new
channel-dependent work begins in Appendix~\ref{app:resource}, where we
evaluate the mixed resources produced by excess noise.

\subsection{The entanglement required at the protocol output}\label{app:tests}\label{app:ree}
The quantity propagated through the converse is regularized relative
entropy of entanglement (REE). For a bipartite normal state, define
\begin{equation}
 \ER(A:B)_\rho=\inf_{\sigma\in\SEP(A:B)}D(\rho\Vert\sigma),
 ~~
 \ERi(\rho)=\lim_{r\to\infty}\frac{\ER(\rho^{\otimes r})}{r}.
 \label{eq:ree}
\end{equation}
Here $\SEP$ is the trace-norm closed convex hull of product states, and
$D(\rho\Vert\sigma)=\Tr[\rho(\log_2\rho-\log_2\sigma)]$ has its usual
support and extended-value interpretation. When $\ER(\rho)<\infty$,
subadditivity ensures that the limit exists and equals the infimum of
$r^{-1}\ER(\rho^{\otimes r})$
\cite{Vedral1997,VedralPlenio1998}. We use the normal-state formulation
recalled in Ref.~\cite[Appendix A 3]{LossAmplifier}.
Both $\ER$ and $\ERi$ are nonincreasing under the admissible LOCC.
For $\ERi$, this follows by applying an LOCC map independently to every
copy before taking the limit in Eq.~\eqref{eq:ree}.

Here a normal state is represented by a positive trace-class operator
of unit trace. The closure in the definition of $\SEP$ matters in
infinite dimensions: restricting to a particular finite or countable
product decomposition would unnecessarily shrink the comparison set
\cite{HSWInfinite}. All bipartitions below are laboratory bipartitions.
For example, $XA:Y$ includes the transmitted system on Alice's side
before a channel call, whereas $X:BY$ includes its output on Bob's side.

The role of REE is to quantify an entanglement resource that cannot
increase under local processing and public communication, in the
standard entanglement-monotone framework
\cite{VidalMonotones,PlenioVirmani,HorodeckiReview}. For an LOCC map
$\Lambda$, every separable comparison $\sigma$ remains separable, and
relative-entropy data processing~\cite{LindbladDPI} gives
$D(\Lambda(\rho)\Vert\Lambda(\sigma))\le D(\rho\Vert\sigma)$.
Taking the infimum proves monotonicity. Subadditivity follows by using
a product of separable comparisons for a product input. Regularization
then allows the comparison for many copies to correlate those copies;
we shall prove additivity only for the specific resource family where
it is needed. None of these definitions identifies REE with an
achievable communication rate.

The relevant operational statement is a lower bound on $\ERi$ at the
protocol output. For an entanglement target, the test is the projection
onto a maximally entangled state of local dimension $d$. For a private
state, it is the corresponding privacy-test projector, which also acts
on the local shields. If a retained state $\zeta$ of finite REE passes
either size-$d$ test with failure probability at most $\delta$, then
\begin{equation}
 \ERi(\zeta)\ge(1-\delta)\log_2d-h_2(\delta).
 \label{eq:test-bound}
\end{equation}
The quantum and private-state versions are proved in
Ref.~\cite[Appendices C 2 and D 2]{LossAmplifier}, respectively; the latter
allows arbitrary separable shield spaces. In this last statement, ``separable'' describes the Hilbert spaces
(they admit countable orthonormal bases), rather than a restriction
that the shield state be unentangled. The shield can contain arbitrary
local auxiliary registers retained by the protocol.

For clarity, a private state of key size $d$ has the form
\begin{equation}
 \gamma=U(\Phi_d\otimes\theta_{S_AS_B})U^\dagger,
 ~~
 U=\sum_{a,b=0}^{d-1}\proj{ab}_{K_AK_B}\otimes U_{ab},
 \label{eq:private-state-explained}
\end{equation}
where $S_A,S_B$ are local shields, $\theta$ is a shield state, and
each $U_{ab}$ is unitary on the joint shield. Measuring the key
registers gives a perfectly correlated uniform key independent of a
purifying eavesdropper~\cite{HHHO}. The associated privacy test is
$\Pi_\gamma=U(\Phi_d\otimes I_{S_AS_B})U^\dagger$.
Its acceptance probability on a separable state is at most $1/d$,
as recalled in the cited companion paper. The untwisted test is simply
$\Phi_d$: on a product vector its overlap is
$d^{-1}|\sum_a x_a y_a|^2\le1/d$ by Cauchy--Schwarz.
This explains why the target size enters the converse through
$\log_2 d$, even when the shields have no dimension bound.

It is useful to recall why the bound involves regularized REE.
Every separable state passes one test with probability at most $1/d$.
More generally, every state separable across the aggregate Alice--Bob
cut passes any specified $\ell$ tests with probability at most
$d^{-\ell}$, even if its different copies are correlated. Test $q$
independent copies of $\zeta$, accepting when at least
$\ell=\lfloor(1-\delta-\epsilon)q\rfloor$ tests pass, where
$0<\epsilon<1-\delta$. Acceptance on $\zeta^{\otimes q}$ tends to one.
For a separable comparison, a union bound gives acceptance at most
$\binom q\ell d^{-\ell}$. Binary relative-entropy data processing
therefore gives the asymptotic lower bound
\[
 \liminf_{q\to\infty}\frac{\ER(\zeta^{\otimes q})}{q}
 \ge (1-\delta-\epsilon)\log_2d-h_2(1-\delta-\epsilon).
\]
Letting $\epsilon\downarrow0$ yields Eq.~\eqref{eq:test-bound}.
The endpoint $\delta=1$ is trivial. This explains the threshold-test
step in the cited proof: a bound on a single test alone would only
control unregularized REE. The tests are mathematical tools for
comparing states and need not themselves be LOCC measurements.

For a quantum target, the entanglement infidelity is precisely the
test failure, so we take $\delta=\varepsilon$. A classical key requires
one additional reduction. The measured key state is a classical state;
we must apply the privacy test to a retained quantum implementation
before the final key measurement, rather than to that measured state.
The construction in
Ref.~\cite[Appendix D 3]{LossAmplifier} retains local purifications,
discarded systems, and Kraus registers. Tracing the added registers
recovers the original protocol at every stage.

We use the squared fidelity
$F(\rho,\sigma)=\norm{\sqrt\rho\sqrt\sigma}_1^2$.
The trace-distance relation
$1-\sqrt{F(\rho,\sigma)}\le\tfrac12\norm{\rho-\sigma}_1$
\cite{FuchsVanDeGraaf}, together with the purification relation for
fidelity~\cite{Uhlmann}, explains the square appearing below.
The retained-implementation construction applies this relation to
the ideal key comparison and its quantum lift. Keeping local
purifying registers is important: the measured classical key alone
is separable and would not provide the quantum state to which the
privacy-test estimate applies.

With the operational trace-distance error $e$ of
Eq.~\eqref{eq:key-error}, the retained state has fidelity at least
$(1-e)^2$ to a private state and consequently passes its privacy test
with probability at least $(1-e)^2$ \cite{HHHO}. Thus
$\delta=1-(1-e)^2=2e-e^2$. Subsequent operations leave the additional
registers unused, so every transmitted optical marginal and every
unconditional mean $x_i$ are unchanged. Public records still reach Eve;
the channel environment and its thermal purification remain in
$E_{\rm all}$. The reduction therefore uses the original energy budget,
without assuming any bound on the shield dimension or its entropy.

\subsection{The resource supplied at a fixed transmitted photon number}\label{app:simulation}
For a block of $s$ transmitted modes with total photon number $m$, let
\begin{equation}
 \cH_{s,m}=\operatorname{span}\{\ket{\bm r}:|\bm r|=m\},~~
 d_{s,m}=\binom{s+m-1}{m},~~
 \ket{\Phi_{s,m}}=\frac{1}{\sqrt{d_{s,m}}}
       \sum_{|\bm r|=m}\ket{\bm r}_R\ket{\bm r}_{A^{\prime s}}.
 \label{eq:sector}
\end{equation}
The constraint fixes the total number, not its distribution among modes:
all occupations $\bm r=(r_1,\ldots,r_s)$ with $\sum_i r_i=m$ are
included. A pure-state symbol without a ket denotes its projector.
For $\cG=\cG_{\tau,v}$, define
\begin{equation}
 \mathcal J_{s,m}=(\id_R\otimes\cG^{\otimes s})(\Phi_{s,m}),
 ~~
 \mathcal J_{s,m}=\bigoplus_{k\ge0}p_k\rho_{s,m,k}.
 \label{eq:J}
\end{equation}
The direct sum resolves Bob's total output number $k$.
The states $\rho_{s,m,k}$ are normalized conditional states; the
probabilities and states depend on $s,m,\tau,v$, which are suppressed
when fixed. Alice's resource share has $m$ photons, while Bob's mean
is $\tau m+sv$. The auxiliary block $A^{\prime s}$ defines the resource;
it is not an additional transmission in the actual communication
protocol.

For example, $s=2$ and $m=1$ give the two-dimensional space spanned
by $\ket{10}$ and $\ket{01}$. Its resource is obtained from
\[
 \ket{\Phi_{2,1}}=
 \frac{\ket{10}_R\ket{10}_{A^{\prime2}}+
       \ket{01}_R\ket{01}_{A^{\prime2}}}{\sqrt2}.
\]
Thus fixing total number leaves a nontrivial quantum system encoded
in the distribution of photons among modes. The input sector is
finite-dimensional, while the noisy output may occupy arbitrarily
many number sectors. This distinction permits an exact finite
measurement at the input without imposing an output photon cutoff
on the simulation itself.

The construction uses the measurement, classical communication, and
conditional correction structure of quantum teleportation
\cite{BennettTeleportation,BraunsteinKimble}. Here the matching
corrections are supplied by the channel's passive-mode covariance
and by the irreducibility of its action on the fixed-number space.
The exactness statement is therefore tied to that input sector.

The simulator is imported from Ref.~\cite[Appendix B 1]{LossAmplifier},
with its full Fock-space output implementation in
Ref.~\cite[Appendix E 1]{LossAmplifier}. Its covariance mechanism comes
from the irreducible-representation construction of Ref.~\cite{GBW}.
We verify the hypotheses needed for the noisy channels here.
The passive transformation $U\in U(s)$ acts irreducibly on the
$m$-photon space through $W_m(U)=\Sym^m U$. The channel satisfies
\begin{equation}
 \cG^{\otimes s}(W_m(U)TW_m(U)^\dagger)
   =\Gamma(U)\cG^{\otimes s}(T)\Gamma(U)^\dagger,
 ~~ \Gamma(U)=\bigoplus_{k\ge0}\Sym^k U.
 \label{eq:covariance}
\end{equation}
Identical thermal environment states are invariant under passive mode
rotations. For an amplifier, the conjugate rotation on its environmental
modes gives the same system covariance. For additive noise, covariance
follows from invariance of the isotropic displacement distribution.
It also follows for all three families from the two quantum-limited
factors in Eq.~\eqref{eq:decomposition}.

The action of the imported simulator can be summarized without
repeating its construction. Alice performs a finite Bell POVM on the
actual input and her resource share. Its normalized Bell vectors are
of the form $(W_m(U_\alpha)\otimes I)\ket{\Phi_{s,m}}$, with positive
weights $w_\alpha$ obeying
\begin{equation}
 \sum_\alpha d_{s,m}^{\,2}w_\alpha
 (W_m(U_\alpha)\otimes I)\Phi_{s,m}
 (W_m(U_\alpha)^\dagger\otimes I)=I.
 \label{eq:recalled-bell-povm}
\end{equation}
Irreducibility supplies this resolution of the identity; its
finite-outcome form is established in the cited reference.
Outcome $\alpha$ produces the desired channel output conjugated by
$\Gamma(U_\alpha)^\dagger$, with weight $w_\alpha$. After Alice announces
$\alpha$, Bob applies $\Gamma(U_\alpha)$. Every corrected branch is
therefore $w_\alpha$ times the desired state. Although the noisy output
contains arbitrarily many photons, $\Gamma(U_\alpha)$ is a bounded
unitary on the full output Fock space.

Consequently the same finite-outcome LOCC map gives
\begin{equation}
 \mathcal T_{s,m}(\rho_{XA^sY}\otimes\mathcal J_{s,m})
   =(\id_X\otimes\cG^{\otimes s}\otimes\id_Y)(\rho_{XA^sY})
 \label{eq:simulation}
\end{equation}
for every input supported on $\cH_{s,m}$, including arbitrary
correlations with memories $X,Y$ on separable Hilbert spaces.
Here $X,A^s$ initially belong to Alice and $Y$ to Bob.
In the simulation, Alice receives the reference share $R$ and Bob the
output share of the independent resource. LOCC monotonicity and
subadditivity imply that
\[
 f_s(m):=\ER(\mathcal J_{s,m})
\]
bounds the entanglement increase at that fixed input number.
This is where the physical input number enters the resource accounting.
The amount $f_s(m)$ must still be evaluated for the noisy resource;
that is the purpose of Appendix~\ref{app:resource}.

The accounting follows directly at the level of a correlated input:
before simulation the resource is independent of $\rho_{XA^sY}$,
so subadditivity bounds the initial REE by
$\ER(XA^s:Y)_\rho+f_s(m)$. The simulator is LOCC across the same
laboratory cut, and tracing out its ancillary registers cannot
increase this quantity. This gives the sector gain bound used in
Appendix~\ref{app:block}. No separability assumption is made on the
original input or on its correlations with either memory.

\subsection{Accounting for photon-number resolution}\label{app:counts}
A general input is not confined to one sector. We therefore need to
control both the average entanglement in its conditional sectors and
the coherence between them. For a local instrument with normalized
conditional states $\rho_j$, selective monotonicity gives
\begin{equation}
 \sum_jp_j\ER(\rho_j)\le\ER(\rho),~~
 \ER\!\left(\bigoplus_jp_j\rho_j\right)=\sum_jp_j\ER(\rho_j).
 \label{eq:flags}
\end{equation}
The second equality applies when the branches are locally
distinguishable, for example because a local register records $j$.
It says that an already classical local flag carries no additional
REE. The countable-outcome statements, with the closed separable set,
are recalled in Ref.~\cite[Appendix A 3]{LossAmplifier}.

The coherence cost is a different quantity. Let
$\cP_B(\omega)=\sum_kP_k^B\omega P_k^B$ denote pinching in Bob's total
output number, and let $K$ be that measurement outcome.
The geometric distribution maximizes entropy among nonnegative
integer distributions with a prescribed mean. Together with the
local-pinching estimate, this gives
\begin{equation}
 H(K)\le g(z)\quad\text{if }\EE K=z,~~
 \ER(\omega)\le\ER(\cP_B\omega)+H(K)
 \label{eq:pinching}
\end{equation}
for a finite-dimensional state in the second inequality.
The results are proved in
Ref.~\cite[Appendices A 4 and A 5]{LossAmplifier}.
The latter inequality bounds the entanglement that can be lost when
coherences between different output numbers are removed. It does not
assert that the measurement can be physically reversed.

The two contributions can be understood separately. If a classical
flag $j$ is already present, the parties can read it locally without
altering the conditional state, which explains the average in
Eq.~\eqref{eq:flags}. Pinching a coherent superposition of sectors
does alter the state. In finite dimensions, a block-diagonal
comparison $\sigma$ obeys
\[
 D(\omega\Vert\sigma)-D(\cP_B\omega\Vert\sigma)
 =S(\cP_B\omega)-S(\omega).
\]
The entropy increase is at most the entropy of the recorded outcome.
This is the origin of the additional $H(K)$ term; it is distinct
from the average entanglement in the flagged state.

The scalar entropy bound also has a direct derivation. For $z>0$,
the geometric law $r_k=(1+z)^{-1}[z/(1+z)]^k$ has mean $z$.
Any other law $p_k$ with that mean has
$D(p\Vert r)=g(z)-H(p)\ge0$. If $z=0$, the count is zero with
certainty. Thus the count-entropy estimate needs only a first
moment, without a bound on its variance.

For a block with mean input $sx$, the total output count has mean
$s(\tau x+v)$. Resolving this one count therefore costs at most
$g(s(\tau x+v))=O(\log s)$ for fixed $x$.
Resolving a separate count in every mode could instead cost a sum of
$s$ entropies and would not yield the needed sublinear correction.
We shall use the single-count inequality before taking any
regularization limit.

We also use the sector-dimension estimate
\begin{equation}
 \log_2d_{s,m}\le s\,g(m/s).
 \label{eq:sector-dimension-entropy}
\end{equation}
It is among the scalar estimates of the same reference.
One way to see its meaning is to consider the maximally mixed state
on $\cH_{s,m}$: its entropy is $\log_2d_{s,m}$, its one-mode means
are all $m/s$, and entropy subadditivity bounds it by the sum of the
thermal entropies at those means.

Appendix~\ref{app:block} first applies these facts at finite input,
output, and memory cutoffs. The output-first procedure of
Ref.~\cite[Appendix E 4]{LossAmplifier} then removes the cutoffs using
lower semicontinuity of unregularized REE
\cite[Corollary 20]{LamiShirokov}.
Only afterwards do replicas remove the count-entropy overhead in
Appendix~\ref{app:replicas}. Keeping these two limits separate is what
allows unrestricted memories without requiring their entropies to be
finite.

\section{Entanglement of noisy photon-sector resources}\label{app:resource}
Throughout this appendix,
$\max\{0,\tau-1\}<v<\tau$.
The resources in Eq.~\eqref{eq:J} have a fixed input number $m$,
but their output number is random. Even after conditioning on an
output number $k$, a noisy channel leaves a mixture of internal
attenuation and amplification processes. The task is therefore to
evaluate the entanglement of the actual mixed state
$\rho_{s,m,k}$, rather than to average the entanglement of
unobserved internal histories.

Use the decomposition
\begin{equation}
 G=1+v,~~ a=\frac{\tau}{G},~~
 \cG_{\tau,v}=\cA_{G,0}\circ\cL_a,~~ 0<a<1<G.
 \label{eq:decomposition}
\end{equation}
The first channel attenuates the mean to $aN$, and the second changes
it to $GaN+G-1=\tau N+v$. The decomposition reproduces the full
Gaussian channel action, not only this mean.

Indeed, the attenuation and amplification covariance maps are
$V\mapsto aV+(1-a)I$ and $V\mapsto GV+(G-1)I$.
Their composition is
$V\mapsto\tau V+(2v+1-\tau)I$, as in
Eq.~\eqref{eq:channel}. Such decompositions and their number-basis
Kraus representations are standard tools for Gaussian channels
\cite{HolevoOneMode,CarusoClassification,IvanKraus}.

No flag recording how many photons were lost and subsequently
created is made available to Alice or Bob. The only physical output
flag used in Eq.~\eqref{eq:J} is $k$.

To make this distinction concrete, let $\ell$ be the photon number
after the auxiliary attenuation. Its distribution for a fixed total
input $m$ is binomial. Conditional on $\ell$, the amplifier adds
photons with a negative-binomial distribution. Summing the internal
label gives the physical output probability
\begin{equation}
 p_k=\frac{1}{G^{s+k}}
 \sum_{\ell=0}^{\min(m,k)}
 \binom m\ell a^\ell(1-a)^{m-\ell}
 \binom{s+k-1}{k-\ell}v^{k-\ell}.
 \label{eq:output-number-probability}
\end{equation}
For each $\ell$, summing over $k\ge\ell$ gives one, and the
conditional output mean is $G\ell+sv$. Consequently
$\sum_kp_k=1$ and $\sum_k k p_k=\tau m+sv$.
The sum over $\ell$ is a way to evaluate a probability, not a
measurement record available to the communicating parties.
Entanglement must still be evaluated on the mixed conditional
state obtained after that sum.

The useful comparison parameter is
\begin{equation}
 z=\frac{a}{(1-a)v}
   =\frac{\tau}{(1+v-\tau)v}>1,~~
 z-1=\frac{(1+v)(\tau-v)}{(1+v-\tau)v}.
 \label{eq:witness-domain}
\end{equation}
Thus $z>1$ is exactly the non-entanglement-breaking condition within
the strictly noisy physical region. In particular, $z$ equals
$\kappa/[(\kappa-1)^2\nu(\nu+1)]$ for noisy amplification and
$1/\xi^2$ for additive noise. Its positivity will be used when
comparing the resource with all separable states.

As in Ref.~\cite[Appendices B 2 and E 2]{LossAmplifier}, take a
unit vector $u\in\mathbb C^s$ and its normalized condensate
\begin{equation}
 \ket{u;r}=\frac{1}{\sqrt{r!}}
       \left(\sum_{i=1}^s u_i a_i^\dagger\right)^r\ket{0^{\otimes s}}
 =\sum_{|\bm n|=r}
       \sqrt{\frac{r!}{\prod_i n_i!}}\prod_i u_i^{n_i}\ket{\bm n}.
 \label{eq:condensate-definition}
\end{equation}
The operators in this definition act on the laboratory in which the
condensate is prepared. The separable comparison is
\begin{equation}
 \sigma_{s,m,k}=
 \int\proj{\bar u;m}_R\otimes\proj{u;k}_B\,\dd u,
 \label{eq:sigma}
\end{equation}
where $\dd u$ is normalized uniform measure on unit vectors.
This integral defines a state in the closed separable set. It is
not an additional continuous-outcome operation in the communication
protocol. The conjugation of Alice's direction matches the
$\overline{\Sym^m U}\otimes\Sym^k U$ symmetry of the resource.

We first prove that this comparison is optimal, including for
products of resources. We then evaluate the entanglement per mode
when both the number of modes and the input number grow.

\subsection{Optimal conditional comparison and product additivity}
Recall $d_{s,r}=\binom{s+r-1}{r}$ and set $d_{s,-1}=0$.
\begin{lemma}[Conditional comparison and additivity]\label{lem:optimal}
For $s\ge2$, the separable state $\sigma_{s,m,k}$ in
Eq.~\eqref{eq:sigma}, defined for every $m,k\ge0$, satisfies
\begin{equation}
 \begin{aligned}
 \ER(\rho_{s,m,k})&=D(\rho_{s,m,k}\Vert\sigma_{s,m,k}),\\
 \ER(\mathcal J_{s,m})&=
       \sum_{k\ge0}p_kD(\rho_{s,m,k}\Vert\sigma_{s,m,k}).
 \end{aligned}
 \label{eq:res-exactJ}
\end{equation}
REE is additive on arbitrary finite products of these conditional
states and of the resources $\mathcal J_{s,m}$.
\end{lemma}
\begin{proof}
The notation in this proof distinguishes eigenvalues from the
probabilities of spectral subspaces:
\begin{center}
\begin{tabular}{ll}
\toprule
Quantity & Meaning at fixed $s,m$\\
\midrule
$k$ & Total output photon number, with probability $p_k$\\
$h$ & Label of a correlation subspace inside the $(m,k)$ sector\\
$D_h$ & Dimension of that subspace (also dependent on $s,m,k$)\\
$\lambda_{k,h}$ & An eigenvalue of the unnormalized block $p_k\rho_{s,m,k}$\\
$q_{k,h}=D_h\lambda_{k,h}$ & Joint probability of the labels $(k,h)$\\
$\mu_h$ & An eigenvalue of the normalized separable comparison\\
\bottomrule
\end{tabular}
\end{center}
In particular, $\sum_hD_h\lambda_{k,h}=p_k$, whereas
$\sum_hD_h\mu_h=1$ for each fixed $k$. Keeping these two
normalizations separate is important when taking logarithms.

The proof has three parts: diagonalize the resource and comparison in
the same number-sector decomposition, certify optimality by a
separable-state witness, and then extend that witness to products
and countable output flags.

\smallskip\noindent\emph{1. Decomposing the fixed-number space.}
Let $r_i,b_i$ annihilate photons in the reference and output modes,
respectively, and put
\[
 C_-=\sum_{i=1}^s r_i b_i,~~ C_+=C_-^\dagger,~~
 A=s+m+k-1.
\]
The operator $C_-$ removes one photon from each laboratory and
contracts their mode indices. Its commutator is
\[
 [C_-,C_+]=s+\hat N_R+\hat N_B.
\]
On a space with fixed numbers $(m_0,k_0)$, this implies
$C_-C_+=C_+C_-+(s+m_0+k_0)I>0$.
Consequently $C_+$ is injective and the adjoint lowering map is onto
the preceding number space. The kernel of $C_-$ at $(m_0,k_0)$
therefore has dimension
$d_{s,m_0}d_{s,k_0}-d_{s,m_0-1}d_{s,k_0-1}$.

For $0\le h\le\min(m,k)$, apply $C_+^h$ to that kernel at
$(m-h,k-h)$ and call the resulting subspace $\Pi_h$.
On a lowest-weight vector $\ket{\zeta}$ with
$C_-\ket{\zeta}=0$, the commutator gives
\begin{equation}
 C_-C_+^h\ket{\zeta}
   =h(s+m_0+k_0+h-1)C_+^{h-1}\ket{\zeta}.
 \label{eq:lowest-weight-action}
\end{equation}
With $m_0=m-h$ and $k_0=k-h$, iteration gives
\begin{align}
 D_h:=\dim\Pi_h
 &=d_{s,m-h}d_{s,k-h}
     -d_{s,m-h-1}d_{s,k-h-1},\label{eq:res-dim}\\
 C_+^jC_-^j\big|_{\Pi_h}
 &=\fall hj\fall{(A-h)}j I.
 \label{eq:res-ladder}
\end{align}
Here $\fall nj=n(n-1)\cdots(n-j+1)$ and $\fall n0=1$.
In particular, the eigenvalue of $C_+C_-$ is $h(A-h)$.
For $0\le h<\min(m,k)$, successive eigenvalues differ by
$A-2h-1>0$. Thus these subspaces are orthogonal.
Their dimensions telescope to $d_{s,m}d_{s,k}$, proving completeness.
The label $h$ is a spectral label of the mixed state; it does not
represent an internal loss or amplification record supplied to a party.

The lowering operator organizes correlations between equal mode
indices in the two laboratories. Starting from a vector annihilated
by $C_-$, repeated application of $C_+$ builds the corresponding
ladder. Its endpoint at the prescribed numbers $(m,k)$ belongs to
exactly one of the subspaces above. The distinct eigenvalues and
the telescoping dimension sum establish the decomposition without
requiring a choice of basis inside each degenerate subspace.

\smallskip\noindent\emph{2. Spectrum of the noisy resource.}
The unnormalized block $p_k\rho_{s,m,k}$ has eigenvalues
\begin{align}
 \lambda_{k,h}
 &=\frac{(1-a)^m v^k}{d_{s,m}G^{s+k}}H_h(z),
 \label{eq:res-eigen}\\
 H_h(z)&=\sum_{j=0}^h\binom hj\binom{A-h}{j}z^j.
 \label{eq:res-H}
\end{align}
We derive this formula by embedding the fixed input sector in a
Gaussian state whose photon-number expansion is explicit.
Choose $0<b<1$ and send one arm of
$\sqrt{1-b}\sum_{n\ge0}b^{n/2}\ket{n,n}$ through the channel.
For $s$ independent copies, the output can be written as
\begin{equation}
 \left(\frac{1-b}{G}\right)^s
 e^{\chi C_+}q_R^{\hat N_R}q_B^{\hat N_B}e^{\chi C_-},
 ~~
 q_R=b(1-a),\quad q_B=\frac vG,\quad
 \chi=\sqrt{\frac{ab}{G}}.
 \label{eq:res-normal}
\end{equation}

To identify this operator, its one-pair Husimi density is
\begin{equation}
 \frac{1-b}{\pi^2G}
 \exp\!\left[-(1-q_R)|\alpha|^2-(1-q_B)|\beta|^2
                  +2\chi\operatorname{Re}(\alpha\beta)\right].
 \label{eq:resource-husimi}
\end{equation}
The precision determinant is
\[
 (1-q_R)(1-q_B)-\chi^2=\frac{1-b}{G}>0.
\]
It both normalizes the Gaussian and gives, after inversion, the
Husimi moments
\[
 \EE|\alpha|^2=\frac1{1-b}=N+1,~~
 \EE|\beta|^2=\tau N+v+1,~~
 \EE(\alpha\beta)=\frac{\sqrt{\tau b}}{1-b}
                =\sqrt{\tau N(N+1)},
 \quad N=\frac b{1-b}.
\]
The extra ones in the diagonal moments are the heterodyne vacuum
noise. These are precisely the moments of the required Gaussian
output, which establishes Eq.~\eqref{eq:res-normal}.

Projecting Alice's $s$ reference modes onto total number $m$ has
probability $d_{s,m}(1-b)^s b^m$. If we also project Bob onto number
$k$, a term with $j$ lowerings on the right must have exactly $j$
raisings on the left. Its scalar coefficient, before normalizing
Alice's event, is
\[
 \left(\frac{1-b}{G}\right)^s
 \frac{\chi^{2j}}{(j!)^2}\,q_R^{m-j}q_B^{k-j}.
\]
Since $\chi^2/(q_Rq_B)=z$, dividing by Alice's probability gives
\begin{equation}
 p_k\rho_{s,m,k}
 =\frac{(1-a)^m v^k}{d_{s,m}G^{s+k}}
   \sum_{j=0}^{\min(m,k)}
       \frac{z^j}{(j!)^2}C_+^jC_-^j.
 \label{eq:res-resource-polynomial}
\end{equation}
The auxiliary parameter $b$ has cancelled, as it must for a resource
defined by a fixed input sector. Substituting
Eq.~\eqref{eq:res-ladder} gives Eqs.~\eqref{eq:res-eigen}--\eqref{eq:res-H}.

\smallskip\noindent\emph{3. Spectrum of the separable comparison.}
On the same subspaces $\Pi_h$, the eigenvalues of
$\sigma_{s,m,k}$ are
\begin{equation}
 \mu_h=\frac{m!k!(s-1)!}{h!(A-h)!}>0.
 \label{eq:res-mu}
\end{equation}
For a direct derivation, consider the heterodyne measure-and-prepare
map of gain $c>0$,
\[
 T\longmapsto
 \int_{\mathbb C^s}\frac{\dd^{2s}\alpha}{\pi^s}
       \langle\alpha|T|\alpha\rangle
       \proj{\sqrt c\,\alpha}.
\]
Its Gaussian parameters are $\tau=v=c$, so in
Eq.~\eqref{eq:decomposition} it has
$G=1+c$, $a=c/(1+c)$, and $z=1$ \cite{MLG}.
Write $\alpha=r u$ with $u$ a unit vector. Projection onto
$(m,k)$ leaves the product directions
$\ket{\bar u;m}\otimes\ket{u;k}$; all dependence on the radius
is scalar. Normalizing the angular integral therefore gives exactly
Eq.~\eqref{eq:sigma}.
The radial Gaussian integral gives the output-count probability
\begin{equation}
 p_k^{\rm het}=\binom{s+m+k-1}{k}
              \frac{c^k}{(1+c)^{s+m+k}}.
 \label{eq:heterodyne-count}
\end{equation}
Vandermonde's identity gives $H_h(1)=\binom Ah$.
Dividing Eq.~\eqref{eq:res-eigen}, evaluated at $z=1$, by
Eq.~\eqref{eq:heterodyne-count} proves Eq.~\eqref{eq:res-mu}.
Thus the resource and comparison commute. Moreover, every $\mu_h$
is positive, so the comparison has full support on the finite
$(m,k)$ space.

\smallskip\noindent\emph{4. Optimality against all separable comparisons.}
It suffices to optimize over states on the local $(m,k)$ subspace.
Indeed, pinch any separable comparison into that subspace and its
complement. This leaves the target $\rho_{s,m,k}$ unchanged and cannot
increase its relative entropy to the comparison. If the comparison
has weight $w>0$ in that subspace, its normalized restriction is
separable and its relative entropy differs by the nonnegative
penalty $-\log_2w$. A zero-weight block gives infinite relative
entropy. We may therefore work entirely in finite dimensions.

Expand $H_h$ around $z=1$. The coefficient of $(z-1)^j$ is
$\binom hj\binom{A-j}{h}$, by Vandermonde's identity, so
\begin{equation}
 \frac{H_h(z)}{\binom Ah}
 =\sum_{j=0}^h
       \frac{\fall hj\fall{(A-h)}j}{j!\fall Aj}(z-1)^j.
 \label{eq:witness-polynomial-expansion}
\end{equation}
Using Eqs.~\eqref{eq:res-ladder} and~\eqref{eq:res-mu}, the ratio
$W=\rho_{s,m,k}\sigma_{s,m,k}^{-1}$ is consequently
\begin{equation}
 W=c_{m,k}\sum_{j=0}^{\min(m,k)}
       \frac{(z-1)^j}{j!\fall Aj}C_+^jC_-^j,
 ~~ c_{m,k}>0.
 \label{eq:res-witness}
\end{equation}
The positive scalar accounts for normalization. Crucially, every
coefficient is nonnegative because $z>1$.

Identify an $m$-photon state with a density operator on
$\Sym^m\mathbb C^s$. For a state $\alpha_R$, let
$\alpha^{(j)}$ denote its normalized $j$-particle reduction, obtained
by tracing out $m-j$ tensor factors; define $\beta^{(j)}$ similarly
for a $k$-photon state. Expanding the ordered annihilation strings in
$C_-^j$ gives
\begin{equation}
 \Tr[(\alpha\otimes\beta)C_+^jC_-^j]
 =\fall mj\fall kj\,
       \Tr[(\alpha^{(j)})^T\beta^{(j)}]
 \le\fall mj\fall kj.
 \label{eq:res-contractbound}
\end{equation}
The falling factorials count the choices of ordered particles.
The transpose is in the induced tensor basis and arises from
contracting equal mode indices between the two laboratories.
The reduced operators are density operators, and their overlap is
at most one. For $j=0$, both reductions are the scalar one.

Aligned conjugate condensates attain equality for every $j$
simultaneously. Their average is $\sigma_{s,m,k}$, and therefore
\begin{equation}
 \sup_{\omega\in\SEP}\Tr(W\omega)
 =\Tr(W\sigma_{s,m,k})=\Tr\rho_{s,m,k}=1.
 \label{eq:witness-normalization}
\end{equation}
Linearity extends the bound from product states to their convex hull
and closure.

Finally, since $\rho$ and $\sigma$ commute and $\sigma$ has full
support, differentiation in any separable direction gives
\begin{equation}
 \left.\frac{\dd}{\dd\epsilon}
 D\bigl(\rho\Vert(1-\epsilon)\sigma+\epsilon\omega\bigr)
 \right|_{\epsilon=0^+}
 =\frac{1-\Tr(W\omega)}{\ln2}\ge0.
 \label{eq:relative-entropy-direction}
\end{equation}
The direction $\omega$ need not commute with $\sigma$:
commutation of $\rho$ with $\sigma$ suffices to reduce the trace
derivative of the matrix logarithm to $\Tr(\rho\sigma^{-1}\omega)$.
Convexity in the second argument then makes this first-order
condition sufficient for global optimality \cite{VollbrechtWerner}.
This proves the first equality in Eq.~\eqref{eq:res-exactJ}.

The witness condition is a supporting-plane certificate for the
relative-entropy minimization. It checks every separable direction
away from $\sigma$, including directions that do not respect the
resource symmetry. Symmetry makes the candidate and its spectrum
tractable; the contraction estimate is what proves optimality over
the full separable set. This distinction is central to the use of
symmetry in entanglement calculations~\cite{VollbrechtWerner}.

\smallskip\noindent\emph{5. Why the same comparison is optimal for products.}
The product claim requires more than testing products independently.
A separable comparison across the aggregate Alice--Bob cut may have
entanglement between blocks within Alice's laboratory or within Bob's.
Let $\alpha_{R_1\cdots R_r}$ and $\beta_{B_1\cdots B_r}$ be arbitrary
such local states. For $\boldsymbol j=(j_1,\ldots,j_r)$, take their
normalized joint particle reductions, retaining $j_\ell$ particles
from block $\ell$. The same contraction identity becomes
\begin{align}
 &\Tr\!\left[(\alpha\otimes\beta)
       \bigotimes_{\ell=1}^r
          C_{\ell,+}^{j_\ell}C_{\ell,-}^{j_\ell}\right]\notag\\
 &\quad=
 \left(\prod_{\ell=1}^r
          \fall{m_\ell}{j_\ell}\fall{k_\ell}{j_\ell}\right)
       \Tr[(\alpha^{(\boldsymbol j)})^T
                            \beta^{(\boldsymbol j)}]
 \le
 \prod_{\ell=1}^r
          \fall{m_\ell}{j_\ell}\fall{k_\ell}{j_\ell}.
 \label{eq:res-joint-reduction}
\end{align}
These joint reductions are density operators even when they do not
factor between blocks. All coefficients in the product of the
expansions~\eqref{eq:res-witness} are nonnegative, and a product of
aligned condensates saturates every term. Hence
$\Tr[(\bigotimes_\ell W_\ell)\omega]\le1$ for every aggregate
separable comparison. Equation~\eqref{eq:relative-entropy-direction}
applies once more, proving optimality of $\bigotimes_\ell\sigma_\ell$.
Relative entropy is additive on that product comparison, which
establishes the asserted conditional REE additivity. No general
additivity property of REE is assumed.

\smallskip\noindent\emph{6. Restoring the output-number flags.}
For one resource, the separable state
$\Sigma=\bigoplus_{k\ge0}p_k\sigma_{s,m,k}$ gives
\[
 \ER(\mathcal J_{s,m})
 \le D(\mathcal J_{s,m}\Vert\Sigma)
 =\sum_{k\ge0}p_kD(\rho_{s,m,k}\Vert\sigma_{s,m,k}).
\]
The probabilities agree in the two direct sums, so their classical
relative-entropy contribution is zero. For the reverse inequality,
Bob resolves finitely many outcomes $k\le K_0$ and one combined
tail outcome. Strong monotonicity bounds the resource REE below by
the partial sum up to $K_0$; the tail contribution is nonnegative.
Letting $K_0\to\infty$ gives equality. Each conditional cost is at
most $\log_2d_{s,m}$, because Alice's dimension is $d_{s,m}$.
Thus this average is finite regardless of the output tail.

For finitely many resources, truncate the number flag in each
resource and use the conditional product result just proved.
The same upper-comparison and lower-partial-sum arguments then
establish additivity for the full resources. For $s=1$, the reference
sector is one-dimensional, so all corresponding resources are
separable and $f_1(m)=0$.
\end{proof}

\paragraph{One-photon example.}
As a check familiar from finite-dimensional entanglement theory,
take $m=k=1$ and $s\ge2$. The one-photon space in each laboratory
is an $s$-dimensional system. There are only two subspaces:
$D_0=s^2-1$ and $D_1=1$, with
$H_0(z)=1$ and $H_1(z)=1+sz$.
The one-dimensional subspace is the maximally entangled state
$\Phi_s$. Normalization gives
\begin{equation}
 \rho_{s,1,1}=F\Phi_s+\frac{1-F}{s^2-1}(I-\Phi_s),
 ~~ F=\frac{1+sz}{s(s+z)}.
 \label{eq:one-photon-isotropic}
\end{equation}
Thus the conditional resource is an isotropic state. The comparison
eigenvalues are $\mu_1=1/s$ and $\mu_0=1/[s(s+1)]$, so its
maximally entangled fraction is $1/s$. In the region $z>1$,
$F>1/s$, and the general formula reduces to
\begin{equation}
 \ER(\rho_{s,1,1})
 =\log_2s-h_2(F)-(1-F)\log_2(s-1).
 \label{eq:one-photon-ree}
\end{equation}
For $s=2$ this is $1-h_2(F)$. The example checks the degeneracy
factors, the conditional normalization, and the separable boundary
$z=1$ in the full calculation. It is an entanglement calculation;
using this conditional state for communication would additionally
require its occurrence probability and an achievable distillation
yield.

\subsection{Asymptotic entanglement per mode}
The exact formula in Lemma~\ref{lem:optimal} is a sum over two labels:
the output number $k$ and the ladder label $h$. In a large block,
both labels concentrate after division by the number of modes.
Their Shannon entropy is sublinear, while the logarithms of the
block dimensions and comparison eigenvalues have nonzero limits
per mode. These limits produce the closed expression in
Eq.~\eqref{eq:F}.

There is a useful statistical-mechanical interpretation of the next
argument. The sector resource has exactly $m_s$ reference photons,
like a microcanonical constraint. A product of squeezed-state outputs
has a fluctuating reference number with a thermal distribution.
Conditioning that product at its mean selects the desired sector
with only a polynomial probability cost. This lets product-state
moment estimates locate the typical spectral labels of the
fixed-number resource. The conditioning does not make the resource
itself a product state.

\begin{lemma}[Sector entanglement density]\label{lem:density}
For every sequence of nonnegative integers $m_s$ with $m_s/s\to N\ge0$,
\begin{equation}
 \frac{\ER(\mathcal J_{s,m_s})}{s}
       \longrightarrow\cF_{\tau,v}(N).
 \label{eq:res-density}
\end{equation}
\end{lemma}
\begin{proof}
\smallskip\noindent\emph{1. Embed the sector at the mean of a product state.}
First suppose $N>0$ and set $N_s=m_s/s$.
Let $\omega_{N_s}$ be the output when one arm of a two-mode squeezed
vacuum of mean $N_s$ is sent through the noisy channel. In
$\omega_{N_s}^{\otimes s}$, Alice's reference marginal is a product
of thermal states. Before the channel, conditioning its total number
$M$ on $m_s$ gives $\Phi_{s,m_s}$, because all occupations with that
total have equal squeezed-state amplitude. Conditioning after the
channel therefore gives $\mathcal J_{s,m_s}$.
The probability is
\begin{equation}
 \begin{aligned}
 \Pr(M=m_s)
 &=d_{s,m_s}\frac{N_s^{m_s}}{(1+N_s)^{s+m_s}}\\
 &\sim[2\pi sN_s(1+N_s)]^{-1/2}.
 \end{aligned}
 \label{eq:res-conditioning-probability}
\end{equation}
The last line is Stirling's formula at the exact mean
$sN_s=m_s$. Allowing the auxiliary mean to depend on $s$ is
deliberate: it prevents the conditioning event from becoming an
exponentially unlikely deviation when $m_s-sN$ grows.

\smallskip\noindent\emph{2. Concentration before conditioning.}
Bob's total number $K$ in the unconditioned product state is a sum
of $s$ independent thermal counts of mean $y_s=\tau N_s+v$. Hence
\[
 \EE(K/s-y_s)^2=\frac{y_s(1+y_s)}s=O(s^{-1}).
\]
For the other spectral observable, put
$c_s=\sqrt{\tau N_s(N_s+1)}$,
$Z_i=r_i b_i-c_s$, and $Y_s=\sum_iZ_i$. Then
\begin{equation}
 C_+C_- -s^2c_s^2
       =s c_s(Y_s+Y_s^\dagger)+Y_s^\dagger Y_s.
 \label{eq:res-centered-ladder}
\end{equation}
Each $Z_i$ has zero mean, and different pairs are independent.
For example, the one-pair Gaussian moments are
\begin{equation}
 \langle Z_i^\dagger Z_i\rangle=N_s y_s,~~
 \langle Z_i Z_i^\dagger\rangle=(N_s+1)(y_s+1),~~
 \langle Z_i^2\rangle=c_s^2.
 \label{eq:centered-single-pair-moments}
\end{equation}
They immediately give
$\langle(Y_s+Y_s^\dagger)^2\rangle=O(s)$.
For the fourth-order expression, expand
$\sum_{i,j,k,\ell}
\langle Z_i^\dagger Z_jZ_k^\dagger Z_\ell\rangle$.
Any term in which an index occurs only once vanishes by centered
independence. The remaining terms contain one or two distinct
indices, so there are $O(s^2)$ of them. Their moments are uniformly
bounded because $N_s$ stays near the fixed positive $N$.
Therefore
$\langle(Y_s^\dagger Y_s)^2\rangle=O(s^2)$.

For self-adjoint operators, $(A+B)^2\le2A^2+2B^2$ as quadratic
forms. Applying it to Eq.~\eqref{eq:res-centered-ladder} and dividing
by $s^4$ gives
\begin{equation}
 \left\langle\left(C_+C_-/s^2-c_s^2\right)^2\right\rangle
       =O(s^{-1}).
 \label{eq:ladder-second-moment}
\end{equation}
All moments used here are finite in the product Gaussian state.

\smallskip\noindent\emph{3. Concentration survives the sector conditioning.}
The operators $M$, $K$, and $C_+C_-$ commute: the latter preserves
the photon number in each laboratory. Their spectral projections
therefore define a joint probability law.
For a fixed positive tolerance, let $D_s$ be a deviation event for
either $K/s$ or $C_+C_-/s^2$ from the values just identified.
Chebyshev's inequality gives $\Pr(D_s)=O(s^{-1})$. Thus
\begin{equation}
 \Pr(D_s\mid M=m_s)
 \le\frac{\Pr(D_s)}{\Pr(M=m_s)}
 =O(s^{-1/2})\longrightarrow0.
 \label{eq:res-conditional-concentration}
\end{equation}
This estimate is the reason for retaining the conditioning
probability in Step 1. Unconditioned concentration alone would
not justify conditioning on an event of vanishing probability.

Under the conditioned resource, let $H$ denote the random ladder
label, so that
$C_+C_-=H(s+m_s+K-1-H)$.
Since $0\le H/s\le\min(N_s,K/s)$, the above concentration implies
\begin{equation}
 K/s\longrightarrow y,~~ H/s\longrightarrow t
 \quad\text{in probability}.
 \label{eq:res-concentration}
\end{equation}
Indeed, every possible limiting value of $H/s$ must solve
$t(1+N+y-t)=\tau N(N+1)$ in $[0,\min(N,y)]$.
The derivative of $h(1+N+y-h)$ on that interval is at least
$1+|N-y|>0$, so the solution is unique and depends continuously
on the other limiting quantities. It is the smaller root in
Eq.~\eqref{eq:rateparams}.

\smallskip\noindent\emph{4. Dimension and comparison eigenvalue per mode.}
Write $p=m_s-h$ and $r=k-h$. The difference of dimensions in
Eq.~\eqref{eq:res-dim} factors as
\begin{equation}
 D_h=d_{s,p}d_{s,r}
 \left[1-\frac{pr}{(s+p-1)(s+r-1)}\right].
 \label{eq:res-dimension-factor}
\end{equation}
At the typical values, $p/s\to u=N-t$ and $r/s\to w=y-t$.
The bracket tends to
$1-uw/[(1+u)(1+w)]>0$.
Its logarithm divided by $s$ therefore vanishes: the subtraction
in $D_h$ does not cancel the leading exponential growth.
Stirling's formula gives
\begin{equation}
 \frac{\log_2D_H}{s}\longrightarrow g(u)+g(w)
 \quad\text{in probability}.
 \label{eq:dimension-density-limit}
\end{equation}

For the comparison eigenvalue,
\[
 -\log_2\mu_h
 =\log_2(h!)+\log_2((A-h)!)
       -\log_2(m_s!)-\log_2(k!)-\log_2((s-1)!).
\]
The coefficients of $s\log_2s$ and of the linear Stirling terms
cancel, since $h+(A-h)=m_s+k+s-1$. The remaining terms give
\begin{equation}
 -\frac{\log_2\mu_H}{s}
 \longrightarrow
 \varphi(t)+\varphi(L-t)-\varphi(N)-\varphi(y),
 \quad L=1+N+y,
 \label{eq:comparison-density-limit}
\end{equation}
again in probability. The contribution corresponding to $s-1$
is $\varphi(1)=0$.

\smallskip\noindent\emph{5. Control the tails before averaging.}
Convergence in probability is not sufficient to take entropy
expectations. We need a bound on the atypical large-$K$ tail.
For successive allowed $h$,
$\mu_{h+1}/\mu_h=(A-h)/(h+1)\ge1$.
Moreover,
$\mu_0^{-1}=d_{s,m_s}d_{s+m_s,k}$
and $D_h\le d_{s,m_s}d_{s,k}$.
Using Eq.~\eqref{eq:sector-dimension-entropy}, we obtain
\begin{align}
 0\le\frac{\log_2D_H}s
 &\le g(N_s)+g(K/s),\label{eq:dimension-tail-bound}\\
 0\le-\frac{\log_2\mu_H}s
 &\le g(N_s)+(1+N_s)
                g\!\left(\frac{K/s}{1+N_s}\right).
 \label{eq:comparison-tail-bound}
\end{align}
Under the resource,
$\EE[K/s]=\tau N_s+v$, so these means are uniformly bounded.
For any nonnegative variables $X_s$ with bounded means,
\begin{equation}
 \EE[g(X_s)\mathbf1_{\{X_s>R\}}]
 \le\left(\sup_{x>R}\frac{g(x)}x\right)
          \sup_s\EE X_s\longrightarrow0.
 \label{eq:res-uniform-integrability}
\end{equation}
The last limit follows from $g(x)=O(\log x)$.
On $X_s\le R$, the function $g$ is bounded. This establishes
uniform integrability of the first bound. The scaled entropy in
Eq.~\eqref{eq:comparison-tail-bound} satisfies the same estimate
uniformly, since $N_s$ is bounded and $1+N_s\ge1$.
Thus Eqs.~\eqref{eq:dimension-density-limit}
and~\eqref{eq:comparison-density-limit} also converge in expectation.
Only a first-moment tail bound is used at this stage.

This tail step is needed because convergence in probability of
$K/s$ and $H/s$ alone does not imply convergence of entropies.
Rare sectors with large photon number might otherwise contribute
to an average logarithm. Equations~\eqref{eq:dimension-tail-bound}
and~\eqref{eq:comparison-tail-bound} control those contributions
by a function growing only logarithmically with the photon density.
The polynomial conditioning cost, the concentration estimates,
and this bound address three separate issues in the density limit.

\smallskip\noindent\emph{6. Assemble the entropy difference.}
The joint spectral probability is
$q_{k,h}=D_h\lambda_{k,h}$: there are $D_h$ equal eigenvalues
$\lambda_{k,h}$ in the $(k,h)$ block.
Writing $\lambda_{k,h}=q_{k,h}/D_h$ gives
\begin{equation}
 S(\mathcal J_{s,m_s})
       =H(K,H)+\EE\log_2D_H.
 \label{eq:resource-entropy-decomposition}
\end{equation}
The optimal flagged comparison has eigenvalues $p_k\mu_h$.
Taking its cross entropy with the resource gives
$H(K)-\EE\log_2\mu_H$. Therefore Lemma~\ref{lem:optimal} implies
\begin{equation}
 \ER(\mathcal J_{s,m_s})
 =-S(\mathcal J_{s,m_s})+H(K)-\EE\log_2\mu_H.
 \label{eq:resource-ree-decomposition}
\end{equation}
Both number-label entropies are negligible per mode:
\[
 H(K,H)\le H(K)+\log_2(m_s+1)
          \le g(\tau m_s+sv)+\log_2(m_s+1)=O(\log s).
\]
Here $H$ has at most $m_s+1$ values, while $K$ has the stated finite
mean. The expectation limits now yield
\begin{align}
 S(\mathcal J_{s,m_s})/s&\longrightarrow g(u)+g(w),\notag\\
 -\EE\log_2\mu_H/s
 &\longrightarrow
 \varphi(t)+\varphi(L-t)-\varphi(N)-\varphi(y).
 \label{eq:res-density-terms}
\end{align}
Their difference is exactly Eq.~\eqref{eq:F}.

The resulting expression has two physically distinct terms. The
quantity $g(u)+g(w)$ is the entropy density of the noisy resource;
the four $\varphi$ terms form its cross-entropy density with the
optimal separable comparison. Their difference is the REE density.
The label entropies $H(K)$ and $H(K,H)$ grow too slowly to survive
division by the number of modes. Although $u,w$ also occur as
symplectic photon numbers of a Gaussian output, this calculation
has optimized over all separable comparisons for the non-Gaussian
fixed-number resource.

Finally, if $N=0$, the preceding Gaussian conditioning argument is
unnecessary. Alice's reference dimension gives
\[
 0\le\frac{\ER(\mathcal J_{s,m_s})}s
       \le\frac{\log_2d_{s,m_s}}s
       \le g(m_s/s)\longrightarrow0=\cF_{\tau,v}(0).
\]
This includes sequences for which $m_s$ grows sublinearly with $s$.
\end{proof}

\section{A common adaptive weak-converse proof}\label{app:converse}
Write $\cE=\cG_{\tau,v}$, $\cF=\cF_{\tau,v}$, and
$f_s(m)=\ER(\mathcal J_{s,m})$.
We first work in the strictly noisy non-entanglement-breaking region,
$\max\{0,\tau-1\}<v<\tau$. Appendix~\ref{app:limits} treats the
entanglement-breaking cases.

The resource-density limit in Lemma~\ref{lem:density} replaces the
quantum-limited energy estimates of
Ref.~\cite[Appendices B 3 and E 3]{LossAmplifier}.
It is an asymptotic statement about one fixed-number resource, so
several steps are needed before applying it to a protocol.
We first obtain a block inequality for arbitrary correlated inputs,
retaining the exact average sector cost. Independent replicas of a
fixed state then convert that average into $\cF(x)$ and remove the
photon-count overhead. Next we prove concavity of $\cF$, and only
then use Jensen's inequality to account for different energies in
successive adaptive transmissions. This order avoids assuming a
noisy-sector bound or concavity before it has been established.

It is helpful to keep the limits in the proof separate:
\begin{center}
\begin{tabular}{p{0.25\textwidth}p{0.31\textwidth}p{0.35\textwidth}}
\toprule
Limit & What is held fixed & Purpose\\
\midrule
Optical and memory cutoffs & A normal input state and a finite block
& Establish the unregularized block inequality on the full spaces\\
Replica number $s\to\infty$ & One input state at one channel use
& Define regularized REE and remove the count-entropy overhead\\
Test-copy number $q\to\infty$ & One final retained protocol state
& Obtain a target-size lower bound on its regularized REE\\
Channel-use number $n\to\infty$ & A sequence of admissible protocols
& Convert the finite-error bounds into a capacity converse\\
\bottomrule
\end{tabular}
\end{center}
The first three are tools for proving inequalities for a fixed finite
protocol. Only the last is the operational communication limit.
In particular, the replica construction supplies no independence
assumption about the inputs occurring at different protocol times.

\subsection{A block inequality for an arbitrary correlated input}\label{app:block}
\begin{lemma}[Noisy-channel block entanglement bound]\label{lem:block}
Let $\rho_{XA^sY}$ have finite $\ER(XA^s:Y)_\rho$ and total
transmitted mean $sx$. Define
$p_m=\Tr(P_m^A\rho)$, where $P_m^A$ projects onto total input number
$m$, and let
$\omega=(\id_X\otimes\cE^{\otimes s}\otimes\id_Y)(\rho)$.
Then
\begin{equation}
 \begin{aligned}
 \ER(X:B^sY)_\omega
 \le{}&\ER(XA^s:Y)_\rho+\sum_mp_m f_s(m)\\
     &+g\bigl(s(\tau x+v)\bigr).
 \end{aligned}
 \label{eq:res-block}
\end{equation}
No dimension or entropy constraint is imposed on the memories $X,Y$.
\end{lemma}
\begin{proof}
The first term on the right side of Eq.~\eqref{eq:res-block}
accounts for entanglement already shared before this block is
transmitted. The second is the average cost of simulating the
conditional photon sectors. The third pays for the coherence
between them. This is the structure needed for an adaptive
converse: it bounds the increase in entanglement, rather than
assuming that the incoming block is independent of either memory.

We use the number-sector reduction from
Ref.~\cite[Appendices B 4 and E 4]{LossAmplifier}.
The simulation and cutoff mechanisms are the same; the resource
cost is now the exact noisy quantity $f_s(m)$.

\smallskip\noindent\emph{1. Relate input and output pinching.}
Initially take finite-dimensional memories and a finite optical input
cutoff. An output cutoff will be understood during the
finite-dimensional entropy step and removed below.
The normalized state conditioned on Alice's total input number is
$\rho_m=P_m^A\rho P_m^A/p_m$ for $p_m>0$.
Let $\omega_m=\cE^{\otimes s}(\rho_m)$, with identity maps on the
memories understood.

Total-number pinching is equivalent to averaging a common phase
rotation on all modes:
\[
 \cP_A(\rho)=\frac1{2\pi}\int_0^{2\pi}
       e^{i\theta\hat N_A}\rho e^{-i\theta\hat N_A}\,\dd\theta.
\]
Phase covariance moves each such rotation through the channel to
Bob's block, so
\begin{equation}
 \cP_B\omega
       =\cE^{\otimes s}(\cP_A\rho)
       =\sum_mp_m\omega_m.
 \label{eq:phase-pinching}
\end{equation}
The identity remains valid with a phase-covariant output cutoff.
It says that output total-number pinching removes the contribution
of coherence between distinct input total numbers. Within each
sector, arbitrary mode correlations and correlations with both
memories remain.

The pinched input is used to compare entanglement quantities.
The actual protocol can retain its original superposition of
different numbers. The output coherence estimate in Step 3
restores an upper bound for that unpinched output, so the proof
does not restrict the sender to number-diagonal states.

\smallskip\noindent\emph{2. Account for every conditional sector.}
For the sector input $\rho_m$, the exact LOCC simulator uses an
independent resource $\mathcal J_{s,m}$.
Its initial laboratory split is
$(XA^sR):(YB^{\prime s})$, where $R$ and $B^{\prime s}$ are the two
resource shares. Subadditivity followed by LOCC monotonicity gives
\begin{equation}
 \ER(X:B^sY)_{\omega_m}
 \le\ER(XA^s:Y)_{\rho_m}+f_s(m).
 \label{eq:sector-gain}
\end{equation}
Applying convexity to Eq.~\eqref{eq:phase-pinching} and selective
monotonicity to Alice's number measurement yields the full chain
\begin{align}
 \ER(\cP_B\omega)
 &\le\sum_mp_m\ER(\omega_m)\notag\\
 &\le\sum_mp_m\ER(\rho_m)+\sum_mp_mf_s(m)\notag\\
 &\le\ER(\rho)+\sum_mp_mf_s(m).
 \label{eq:block-chain}
\end{align}
The states $\rho_m$ need not be separable. Their pre-existing
entanglement is accounted for by the first term, rather than
discarded from the argument.

\smallskip\noindent\emph{3. Bound the coherence removed at the output.}
Equation~\eqref{eq:pinching} gives
\begin{equation}
 \ER(\omega)\le\ER(\cP_B\omega)+H(K),~~
 H(K)\le g\bigl(s(\tau x+v)\bigr),
 \label{eq:block-overhead}
\end{equation}
where $K$ is Bob's single total output count. Combining this with
Eq.~\eqref{eq:block-chain} proves Eq.~\eqref{eq:res-block} in finite
dimensions. The coherence is restored only in the inequality:
the proof does not require an operation that undoes the pinching.
The overhead is one count entropy, not a sum of separate
mode-count entropies.

\smallskip\noindent\emph{4. Remove the output cutoff first.}
Noise can produce arbitrarily many output photons even when the
input support is finite. Use the local vacuum-replacement channel
\begin{equation}
 \mathcal C_J(T)=P_{\le J}TP_{\le J}
       +\Tr[(I-P_{\le J})T]\proj{0^{\otimes s}},
 ~~ P_{\le J}=\sum_{k=0}^J P_k,
 \label{eq:cutoff}
\end{equation}
on Bob's optical block. It is trace preserving, commutes with total
phase rotations and number pinching, and cannot increase the
output mean. Appending it to a sector simulator does not increase
the resource cost. Hence the same finite-dimensional argument
applies to
$\omega^{(J)}=(\id_{XY}\otimes\mathcal C_J)(\omega)$,
with a right-hand side bounded uniformly by that in
Eq.~\eqref{eq:res-block}.

For a fixed input and fixed memory cutoffs,
$\omega^{(J)}\to\omega$ in trace norm as $J\to\infty$.
Lower semicontinuity of unregularized REE
\cite[Corollary 20]{LamiShirokov} gives
\[
 \ER(\omega)\le\liminf_{J\to\infty}\ER(\omega^{(J)}).
\]
Thus the block bound already holds with infinite output support
before we remove the input or memory cutoffs.

\smallskip\noindent\emph{5. Remove the input and memory cutoffs.}
Now use the local finite-rank memory cutoffs and input vacuum
replacement of
Ref.~\cite[Appendix B 4]{LossAmplifier}.
Each memory cutoff is a local trace-preserving map and leaves the
optical marginal unchanged. The input cutoff
$\mathcal C_L$ preserves $p_m$ for $0<m\le L$ and sends all
discarded probability to $m=0$.
Since $f_s(0)=0$, the new average resource cost is at most the
original one. The mean also decreases, and the initial REE cannot
increase under these local maps.

The original resource average is finite:
\begin{equation}
 0\le\sum_mp_m f_s(m)
 \le\sum_mp_m\log_2d_{s,m}
 \le s\sum_mp_m g(m/s)
 \le sg(x)<\infty.
 \label{eq:resource-domination}
\end{equation}
The last step uses concavity of the known thermal entropy $g$,
not concavity of the new function $\cF$.
Consequently every truncated-input bound has a right-hand side
no larger than the desired original bound.
The input and memory cutoff states converge in trace norm.
Channel contractivity and a second application of lower
semicontinuity complete the limit.

The local cutoff is trace preserving because the discarded weight
is replaced by vacuum. Hence it does not condition on a successful
truncation event or divide the energy by a success probability.
For a normal state, finite-rank memory projections approach the
identity in trace norm, even if the memory entropy is infinite.
The optical cutoff additionally decreases the photon mean.
Lower semicontinuity transfers the finite-dimensional inequality
to the limit while retaining its upper bound. These observations
explain why no memory entropy appears in the block estimate.

This is the output-first version of the amplifier cutoff argument
in Ref.~\cite[Appendix E 4]{LossAmplifier}, with the noisy average
resource cost controlled by Eq.~\eqref{eq:resource-domination}.
No memory entropy has been evaluated. All cutoff limits have been
taken for unregularized REE, so none is exchanged with the
regularization limit that follows.
\end{proof}

\subsection{Replicas remove the photon-count overhead}\label{app:replicas}
We next use the replica method of
Ref.~\cite[Appendix B 5]{LossAmplifier}, also used for amplification
in its Appendix E 5. The replicas here are independent copies of
one fixed input state. This supplies a large optical block without
assuming independence between different times in an adaptive
protocol.

For those replicas, the photon-count law is the product of the
diagonal number distributions of $\rho_A$. Off-diagonal matrix
elements of $\rho_A$ do not affect this classical law, although
their effect on the output has already been covered by the block
coherence term. The law of large numbers therefore uses exactly
the finite mean $x$ required by the operational energy constraint.
No Gaussian form or finite photon-number variance is needed.

\begin{lemma}[One-use regularized entanglement gain]\label{lem:gain}
For an input $\rho_{XAY}$ with
$\ER(XA:Y)_\rho<\infty$ and
$x=\Tr(\hat n_A\rho_A)<\infty$, let
$\omega=(\id_X\otimes\cE\otimes\id_Y)(\rho)$. Then
\begin{equation}
 \ERi(X:BY)_\omega
       \le\ERi(XA:Y)_\rho+\cF(x).
 \label{eq:gain-proof}
\end{equation}
\end{lemma}
\begin{proof}
\smallskip\noindent\emph{1. Apply the block bound to copies of the fixed input.}
Take $s$ copies of $\rho$, regard $A^s$ as the transmitted block,
and regard $X^s,Y^s$ as its memories. If $M_s$ is their total input
photon number, Lemma~\ref{lem:block} gives
\begin{equation}
 \begin{aligned}
 \frac{\ER(\omega^{\otimes s})}{s}
 \le{}&\frac{\ER(\rho^{\otimes s})}{s}
       +\EE\frac{f_s(M_s)}s\\
     &+\frac{g(s(\tau x+v))}{s}.
 \end{aligned}
 \label{eq:replica-bound}
\end{equation}
The distribution of $M_s$ is the distribution of a sum of $s$
independent copies of the photon count of $\rho_A$.
This statement holds even when $\rho_A$ is not diagonal in number:
the number probabilities of the tensor product factor.
The summands have finite mean $x$; no variance assumption is made.

\smallskip\noindent\emph{2. Apply the deterministic density lemma to a random count.}
The first-moment law of large numbers gives $M_s/s\to x$ in
probability. Lemma~\ref{lem:density} says that
$f_s(m_s)/s\to\cF(x)$ for \emph{every} deterministic sequence
$m_s/s\to x$. This has the following local uniform consequence:
for every $\epsilon>0$, there exist $\delta>0$ and $s_0$ such that
\begin{equation}
 s\ge s_0,\quad |m/s-x|<\delta
 \quad\Longrightarrow\quad
 \left|\frac{f_s(m)}s-\cF(x)\right|<\epsilon.
 \label{eq:local-uniform-density}
\end{equation}
To verify this implication, suppose it failed. One could then choose
$s_j\to\infty$ and integers $m_j$ with
$|m_j/s_j-x|<1/j$, while the displayed error remains at least
$\epsilon$. Fill the remaining indices with $m_s=\lfloor sx\rfloor$.
The resulting sequence would contradict Lemma~\ref{lem:density}.

Combining Eq.~\eqref{eq:local-uniform-density} with the law of large
numbers gives $f_s(M_s)/s\to\cF(x)$ in probability.
For $x=0$, the input count is zero almost surely, and the same
conclusion follows directly from $f_s(0)=\cF(0)=0$.

\smallskip\noindent\emph{3. Average using only the first moment.}
The resource bound
$0\le f_s(M_s)/s\le g(M_s/s)$ controls the random tail:
\begin{equation}
 \EE\!\left[g(M_s/s)\mathbf1_{\{M_s/s>R\}}\right]
 \le x\sup_{z>R}\frac{g(z)}z\longrightarrow0.
 \label{eq:replica-ui}
\end{equation}
On $M_s/s\le R$, the dominating entropy is bounded.
Thus the resource densities are uniformly integrable, and
convergence in probability can be upgraded to convergence of
their averages:
\begin{equation}
 \lim_{s\to\infty}\EE\frac{f_s(M_s)}s=\cF(x).
 \label{eq:averaged-density}
\end{equation}
This is the new averaging step needed for noisy channels. It uses
the sublinear growth of $g$, and remains valid for photon
distributions with infinite variance.

\smallskip\noindent\emph{4. Take the regularization limit.}
The input unregularized REE is finite by hypothesis.
The output unregularized REE is already finite from
Lemma~\ref{lem:block} at $s=1$.
The first two REE terms in Eq.~\eqref{eq:replica-bound} therefore
converge to their regularizations, while
$g(s(\tau x+v))/s\to0$.
Using Eq.~\eqref{eq:averaged-density} proves
Eq.~\eqref{eq:gain-proof}, equivalently Eq.~\eqref{eq:gain-main}.
\end{proof}

The replica parameter $s$ is an auxiliary proof parameter.
For an adaptive protocol, we apply Lemma~\ref{lem:gain} to the
actual unconditional state at each fixed channel use, with the
protocol held fixed. The replicas do not replace successive
adaptive inputs by independent ones.

\subsection{Why the energy function is concave}\label{app:concavity}
The gain lemma bounds one use at its actual mean input energy.
To combine uses with different means, we still need concavity of
$\cF$. Instead of differentiating the closed formula, we prove
concavity by preparing a finite ensemble of Gaussian inputs.
The essential additional fact is that the regularized REE of the
flagged output ensemble equals the average of its individual
densities. We establish that fact below, rather than assume a
general additivity or flags property for regularized REE.

This step should be distinguished from a general Gaussian
extremality argument~\cite{WolfExtremality}. The input of the
one-use gain lemma is still arbitrary. Gaussian states enter
here because their number-sector decomposition gives a family
whose regularized REE can be computed exactly. That family then
tests the dependence of the bound on energy. The proof derives
the needed equality for flagged mixtures rather than presuming
a general flags identity for regularized REE.

\begin{lemma}[Gaussian-state density and concavity]\label{lem:concavity}
Let $\omega_N$ be the output obtained by sending one arm of a two-mode
squeezed vacuum of mean $N$ through $\cE$. Then
\begin{equation}
 \ERi(\omega_N)=\cF(N).
 \label{eq:gaussian-density}
\end{equation}
The function $\cF$ is concave and nondecreasing.
\end{lemma}
\begin{proof}
\smallskip\noindent\emph{1. Bound a product of groups of Gaussian outputs.}
For a finite list of energies $N_i$ and nonnegative integer group
sizes $n_i$, define
\[
 \eta_{\boldsymbol n}=\bigotimes_{i:n_i>0}
                          \omega_{N_i}^{\otimes n_i}.
\]
Set
\begin{equation}
 A(\boldsymbol n)=\sum_{i:n_i>0}\EE f_{n_i}(M_i).
 \label{eq:group-average}
\end{equation}
Here $M_i$ is Alice's total reference number in group $i$.
Pinching that number in every group turns the corresponding
squeezed inputs into maximally entangled number-sector inputs.
The conditional outputs are therefore products of
$\mathcal J_{n_i,m_i}$, with the numbers locally distinguishable
on Alice's side.
The unregularized flags identity and the product additivity from
Lemma~\ref{lem:optimal} give
$\ER(\cP\eta_{\boldsymbol n})=A(\boldsymbol n)$.
LOCC monotonicity proves the lower inequality in
\begin{equation}
 \begin{aligned}
 A(\boldsymbol n)
 &\le\ER\!\left(\bigotimes_i\omega_{N_i}^{\otimes n_i}\right)\\
 &\le A(\boldsymbol n)+\sum_{i:n_i>0}H(M_i).
 \end{aligned}
 \label{eq:res-group-sandwich}
\end{equation}

\smallskip\noindent\emph{2. Upper inequality and its finiteness.}
For each sector resource take the optimal comparison furnished by
Lemma~\ref{lem:optimal}. For the group product, take the product of
these comparisons conditional on all $m_i$, and then their direct
sum $\Sigma$ with the actual number probabilities.
This is separable and block diagonal under the group pinching.
Its relative entropy from the pinched state is
\[
 D(\cP\eta_{\boldsymbol n}\Vert\Sigma)=A(\boldsymbol n).
\]
The Gaussian state $\eta_{\boldsymbol n}$ has finite entropy.
The entropy increase under these number pinchings is bounded by
$\sum_iH(M_i)\le\sum_i g(n_iN_i)<\infty$.
Also $A(\boldsymbol n)\le\sum_i n_i g(N_i)<\infty$.
Thus the comparison cross entropy is finite:
\begin{equation}
 -\Tr[(\cP\eta_{\boldsymbol n})\log_2\Sigma]
 =S(\cP\eta_{\boldsymbol n})+A(\boldsymbol n)<\infty.
 \label{eq:group-finite-cross-entropy}
\end{equation}
Since $\log_2\Sigma$ is block diagonal, its expectation is unchanged
by pinching. Consequently
\begin{equation}
 D(\eta_{\boldsymbol n}\Vert\Sigma)
 =A(\boldsymbol n)+S(\cP\eta_{\boldsymbol n})
                         -S(\eta_{\boldsymbol n})
 \le A(\boldsymbol n)+\sum_iH(M_i).
 \label{eq:group-cross-entropy-identity}
\end{equation}
This proves the upper inequality in
Eq.~\eqref{eq:res-group-sandwich}.
One may first resolve finitely many number outcomes and one tail;
the finite entropy and cross-entropy bounds justify the countable
limit. In particular, no subtraction of infinite entropies is
involved.

For a single group of size $s$, divide
Eq.~\eqref{eq:res-group-sandwich} by $s$.
Equation~\eqref{eq:averaged-density} identifies the averaged
resource term, while $H(M)/s\le g(sN)/s\to0$.
This proves Eq.~\eqref{eq:gaussian-density}.

\smallskip\noindent\emph{3. Compute the density of a classically flagged ensemble.}
Let $p_i$ be a finite probability distribution and retain its flag
at both laboratories:
\begin{equation}
 \Omega=\sum_i p_i\proj i_{C_A}\otimes\proj i_{C_B}
                             \otimes\omega_{N_i}.
 \label{eq:flagged-gaussian}
\end{equation}
In $\Omega^{\otimes s}$, condition on the entire classical flag
sequence. Local permutations regroup equal flags, giving group
sizes $n_i$ with $n_i/s\to p_i$ in probability.
Conditional unregularized REE is the REE of
$\eta_{\boldsymbol n}$, so the flags identity lets us average
Eq.~\eqref{eq:res-group-sandwich} over these group sizes.
Its overhead is bounded uniformly by
$\sum_i g(sN_i)=O(\log s)$.

To justify the other average explicitly, put
$b_i(r)=r^{-1}\EE f_r(M_i)$ for $r\ge1$, where here $M_i$ is the
total count in a group of size $r$ at energy $N_i$.
Equation~\eqref{eq:averaged-density} gives
$b_i(r)\to\cF(N_i)$, and
$0\le b_i(r)\le g(N_i)$.
Whenever $p_i>0$, $n_i$ tends to infinity in probability, and thus
\[
 \frac{A(\boldsymbol n)}s
 =\sum_{i:n_i>0}\frac{n_i}s\,b_i(n_i)
 \longrightarrow\sum_i p_i\cF(N_i).
\]
The finite list of bounds $g(N_i)$ controls the expectation of this
expression. Terms with $p_i=0$ are absent. Dividing the averaged
sandwich by $s$ therefore yields
\begin{equation}
 \ERi(\Omega)=\sum_i p_i\cF(N_i).
 \label{eq:res-flagged-density}
\end{equation}
The flags in Eq.~\eqref{eq:flagged-gaussian} were classical from the
outset. Their conditioning does not remove quantum coherence, and
there is no additional penalty $sH(p)$.

\smallskip\noindent\emph{4. Apply the one-use gain bound to that ensemble.}
Alice samples $i$, communicates the classical flag, and locally
prepares a two-mode squeezed vacuum of mean $N_i$.
Both optical arms initially belong to Alice, so this flagged input
has zero entanglement across the laboratory cut.
Sending one arm through $\cE$ produces $\Omega$ with unconditional
input mean $\bar N=\sum_i p_iN_i$.
Lemma~\ref{lem:gain} and Eq.~\eqref{eq:res-flagged-density} give
\begin{equation}
 \sum_i p_i\cF(N_i)=\ERi(\Omega)
       \le\cF\!\left(\sum_i p_iN_i\right).
 \label{eq:concavity}
\end{equation}
This is concavity. The gain lemma was obtained without using it,
so the argument is not circular.

The direction of the inequality has a simple preparation
interpretation. A classical choice among energies $N_i$ can be
made locally, and its mean transmitted energy is $\bar N$.
The parties keep the choice as an accessible classical record.
The flagged output has the average entanglement density on the
left, whereas the one-use gain bound constrains it by the function
evaluated at $\bar N$. A mixture of thermal inputs need not itself
be thermal: the retained flag is what makes this comparison
possible. General energy averaging in an adaptive protocol is
applied only after this concavity statement has been established.

Nonnegativity and finiteness follow from the resource-density limit,
which also gives $\cF(N)\le g(N)$ and $\cF(0)=0$.
Finally, a finite concave nonnegative function on the entire
half-line cannot have a negative secant slope.
If $\cF(b)<\cF(a)$ for $0\le a<b$, then for every $z>b$ concavity
would give
\[
 \frac{\cF(z)-\cF(b)}{z-b}
 \le\frac{\cF(b)-\cF(a)}{b-a}<0.
\]
The resulting linear upper bound would become negative for large
$z$, contradicting nonnegativity. Hence $\cF$ is nondecreasing.
\end{proof}

\subsection{From one use to a fully adaptive protocol}\label{app:adaptive}
We now use the adaptive accounting framework of
Ref.~\cite[Appendices B 5 and E 5]{LossAmplifier}.
The previous lemmas supply its channel-dependent gain function,
including the finiteness and concavity properties needed here.

Fix an admissible protocol with $n$ channel uses.
Immediately before use $i$, write its actual unconditional state
as $\rho^{(i)}_{X_iA_iY_i}$.
Alice owns $X_iA_i$, Bob owns $Y_i$, and $A_i$ is the next mode
to be sent. The memories include all local quantum registers and
all active classical histories.
Thus feedback and correlations between histories and transmitted
states are already represented in $\rho^{(i)}$.
Tracing over Eve gives this laboratory state; it is not conditioned
on a selected successful history.

For example, if $h$ labels the public history before use $i$,
then its optical marginal and mean are
\begin{equation}
 \rho^{(i)}_{A_i}=\sum_h p_h\rho^{(i,h)}_{A_i},
 ~~
 x_i=\sum_h p_h\Tr(\hat n_{A_i}\rho^{(i,h)}_{A_i}).
 \label{eq:history-energy-average}
\end{equation}
The conditional energies may vary strongly with $h$; the bound
uses their unconditional average. The full state still keeps the
correlations between the history, the transmitted mode, and the
memories. This is why applying the one-use estimate does not
discard the advantage potentially supplied by feedback.

Let $\omega^{(i)}$ be the state immediately after the channel call
and let $x_i=\Tr(\hat n_{A_i}\rho^{(i)}_{A_i})$.
Lemma~\ref{lem:gain} gives
\begin{equation}
 \ERi(X_i:B_iY_i)_{\omega^{(i)}}
 \le\ERi(X_iA_i:Y_i)_{\rho^{(i)}}+\cF(x_i).
 \label{eq:one-use-adaptive}
\end{equation}
Before the next use, arbitrary allowed LOCC, local preparation,
and local discarding cannot increase $\ERi$.
The pre-use entanglement at step $i+1$ is therefore at most the
post-use value at step $i$.
Starting with zero entanglement and iterating gives
\begin{equation}
 \ERi(\zeta_n)\le\sum_{i=1}^n\cF(x_i),
 \label{eq:telescope}
\end{equation}
where $\zeta_n$ is the final retained state after the last LOCC.

There is one domain condition to check before using the gain lemma
repeatedly: its input must have finite unregularized REE. The
following argument establishes this from the block inequality,
so it does not assume the conclusion of the adaptive telescope.

For completeness, all intermediate states meet the finite
unregularized-REE hypothesis of Lemma~\ref{lem:gain}.
At $s=1$, the reference number sectors are one-dimensional, so
$f_1(m)=0$. Lemma~\ref{lem:block} therefore bounds the one-use
unregularized increase by $g(\tau x_i+v)<\infty$.
Every $x_i$ is finite because the nonnegative means obey the finite
total budget. Induction from the initial separable state, with LOCC
monotonicity between calls, proves the required finiteness.
This also applies to the larger local memories of a retained key
implementation.

At this point concavity and monotonicity, now proved, can be used:
\begin{equation}
 \frac1n\sum_{i=1}^n\cF(x_i)
 \le\cF\!\left(\frac1n\sum_{i=1}^n x_i\right)
 \le\cF(N).
 \label{eq:energy-average}
\end{equation}
The first inequality combines the possibly different unconditional
energies, and the second uses the actual constraint in
Eq.~\eqref{eq:energy}. This proves Eq.~\eqref{eq:adaptive}.
No branchwise energy restriction or postselection normalization
has been substituted for that constraint.

\subsection{Conversion to operational rate bounds}\label{app:operational}
The remaining step is to compare the entanglement available in
Eq.~\eqref{eq:telescope} with that required by a successful target.
We use the reductions from
Ref.~\cite[Appendices C 2, D 2, and D 3]{LossAmplifier}, recalled
in Appendix~\ref{app:tests}.

For quantum communication or entanglement distribution, apply the
size-$d_n$ maximally entangled-state test to the final decoded
registers. Local decoding and discarding cannot increase REE, so
Eq.~\eqref{eq:test-bound} with $\delta=\varepsilon_n$ gives
\[
 (1-\varepsilon_n)\log_2d_n-h_2(\varepsilon_n)
 \le\ERi(\zeta_n)
 \le\sum_i\cF(x_i).
\]
This is Eq.~\eqref{eq:quantum-error}.

For a key protocol, first use its input-preserving retained
implementation. Its optical marginals, hence every $x_i$, are the
same as in the original protocol. Its final retained quantum
state passes a size-$d_n$ privacy test with failure at most
$\delta_n=2e_n-e_n^2$.
Applying Eqs.~\eqref{eq:test-bound} and~\eqref{eq:telescope} now gives
\[
 (1-e_n)^2\log_2d_n-h_2(2e_n-e_n^2)
 \le\ERi(\zeta_n)
 \le\sum_i\cF(x_i),
\]
which is Eq.~\eqref{eq:key-bound}.
This proves Lemma~\ref{lem:finite} with the quantum and key
error conventions of Sec.~\ref{sec:model}.

Finally put $R_n=n^{-1}\log_2d_n$.
For a vanishing-error quantum sequence, the preceding estimates give
\[
 (1-\varepsilon_n)R_n\le\cF(N)+\frac{h_2(\varepsilon_n)}n.
\]
For a key sequence the left coefficient is $(1-e_n)^2$ and the
entropy term is $h_2(2e_n-e_n^2)/n$.
In either case the coefficient tends to one and the entropy term
vanishes because $h_2\le1$. Thus every achievable vanishing-error
rate is at most $\cF(N)$.
The input-preserving task conversions recalled in
Ref.~\cite[Appendix A 2]{LossAmplifier} give the full capacity
statement of Theorem~\ref{thm:gaussian}.
The replica limit establishing $\ERi$ was taken for a fixed state,
and hence for each fixed finite protocol, before this
communication limit in $n$.

For fixed nontrivial errors, the same inequalities also have a
direct finite-blocklength interpretation. The quantities
$\sum_i\cF(x_i)$ and the binary-entropy correction bound the
permitted target size after multiplication by the corresponding
success coefficient. In the capacity limit, $h_2\le1$ makes
the correction vanish per channel use, while $1-\varepsilon_n$
and $(1-e_n)^2$ approach one. This is the precise point where
the finite-error lemma becomes the weak-converse theorem.

\section{Channel limits and benchmark calculations}\label{app:benchmarks}
This appendix connects the common formula to its boundary values and
to the numerical comparisons in the main text. The limits keep the
channel parameters fixed unless explicitly stated otherwise.

The three channel families share the same resource formula, but
their large-energy symplectic spectra have different scalings.
For thermal loss, $u/N\to1-\eta$ and $w\to\nu$; for noisy
amplification, $u\to\nu$ and $w/N\to\kappa-1$.
For additive noise, both $u$ and $w$ grow as $\sqrt{\xi N}$.
The calculations below explain how these different cancellations
produce the appropriate PLOB limit in each case.

\subsection{Zero energy and the entanglement-breaking boundary}\label{app:limits}
At $N=0$, Eq.~\eqref{eq:rateparams} gives
$y=v$, $L=1+v$, and $(t,u,w)=(0,0,v)$.
The cross-entropy part of Eq.~\eqref{eq:F} is then
$\varphi(1+v)-\varphi(v)=g(v)$, which cancels $g(w)$.
Thus $\cF_{\tau,v}(0)=0$.
Operationally, a zero unconditional mean for a nonnegative photon
number operator forces every transmitted mode to be vacuum.
Output thermal photons do not by themselves distribute
entanglement or secret key.

To examine the entanglement-breaking boundary, use the root equation
to write
\[
 uw=(N-t)(y-t)
    =Ny-t(N+y)+t^2
    =t+N(v-\tau).
\]
At $v=\tau$, this becomes $uw=t$. Consequently
\begin{equation}
 N=u(1+w),~~
 y=w(1+u),~~
 L-t=(1+u)(1+w).
 \label{eq:eb-identity}
\end{equation}
Substitution into the cross-entropy expression gives
\begin{align}
 &\varphi(uw)+\varphi((1+u)(1+w))
       -\varphi(u(1+w))-\varphi(w(1+u))\notag\\
 &~~=(1+u)\log_2(1+u)-u\log_2u
        +(1+w)\log_2(1+w)-w\log_2w
        =g(u)+g(w).
 \label{eq:eb-cancellation}
\end{align}
The equality follows by expanding the logarithm of each product.
Terms with a zero factor are understood by continuity.
It proves that the non-entanglement-breaking formula approaches
zero continuously as $v\uparrow\tau$.

For $v\ge\tau$, the channel is entanglement breaking
\cite{HolevoWerner,HolevoEB,PLOB,MLG}.
Its measure-and-prepare realization cannot increase entanglement
across the laboratory cut. Starting from local preparations and
public randomness, an adaptive protocol therefore has no final
laboratory entanglement to pass a nontrivial asymptotically
error-free target test. The same statement applies to a retained
key implementation, so the privacy-test bound also excludes a
positive key rate. This establishes the cases in which
$\cF_{\tau,v}$ is defined to be zero, without extending the
$z>1$ witness argument outside its domain.

\subsection{Infinite-energy limits for the three noisy families}\label{app:highenergy}
For a fixed non-entanglement-breaking channel, separate the two
terms of the bound as
\begin{equation}
 C_N=\varphi(t)+\varphi(L-t)-\varphi(N)-\varphi(y),
 ~~ S_N=g(u)+g(w),~~ \cF(N)=C_N-S_N.
 \label{eq:cross-and-entropy}
\end{equation}
The thermal entropy satisfies
$g(x)=\log_2x+\log_2e+o(1)$ as $x\to\infty$.
The limiting behavior of $u,w$ can be read directly from their
quadratic equations. Substituting $t=N-u$ or $t=y-w$ into the
root equation gives
\begin{align}
 u^2+[1+v+(\tau-1)N]u&=N(1+v-\tau),
 \label{eq:u-quadratic}\\
 w^2+[1-v+(1-\tau)N]w&=v(N+1).
 \label{eq:w-quadratic}
\end{align}
These equations also explain the different scaling at unit gain.

\smallskip\noindent\emph{Thermal loss.}
Let $\tau=\eta<1$ and $v=(1-\eta)\nu$.
In Eq.~\eqref{eq:w-quadratic}, the coefficient of $w$ grows as
$(1-\eta)N$, while the right-hand side grows as $vN$.
The nonnegative root is bounded and approaches
$v/(1-\eta)=\nu$. Expanding that root one order further gives
$w=\nu+O(N^{-1})$. It follows that
\begin{equation}
 t=\eta N-\eta\nu+O(N^{-1}),\quad
 u=(1-\eta)N+\eta\nu+O(N^{-1}),\quad
 w=\nu+O(N^{-1}).
 \label{eq:loss-root-limit}
\end{equation}

For the cross entropy, it is convenient to pair nearby arguments:
$t-y=-w$ and $L-t=N+1+w$. Therefore
\[
 C_N=[\varphi(y-w)-\varphi(y)]
       +[\varphi(N+1+w)-\varphi(N)].
\]
For $x\to\infty$ and a fixed displacement $b$,
$\varphi(x+b)-\varphi(x)
=b(\log_2x+\log_2e)+o(1)$.
Using $y\sim\eta N$ and $w\to\nu$ gives
\begin{align}
 C_N&=-\nu(\log_2(\eta N)+\log_2e)
       +(1+\nu)(\log_2N+\log_2e)+o(1)\notag\\
    &=\log_2N+\log_2e-\nu\log_2\eta+o(1),\\
 S_N&=\log_2N+\log_2(1-\eta)+\log_2e+g(\nu)+o(1).
\end{align}
Their difference is the thermal-loss PLOB expression in
Eq.~\eqref{eq:loss-plob}.

\smallskip\noindent\emph{Noisy amplification.}
Let $\tau=\kappa>1$ and $v=(\kappa-1)(\nu+1)$.
Now Eq.~\eqref{eq:u-quadratic} has a coefficient growing as
$(\kappa-1)N$ and right-hand side $(\kappa-1)\nu N$.
Its nonnegative root is $u=\nu+O(N^{-1})$, and hence
\begin{equation}
 t=N-\nu+O(N^{-1}),~~
 u=\nu+O(N^{-1}),~~
 w=(\kappa-1)(N+1)+\kappa\nu+O(N^{-1}).
 \label{eq:amp-root-limit}
\end{equation}
Using $L-t=y+1+u$ instead pairs the cross-entropy terms as
\[
 C_N=[\varphi(N-u)-\varphi(N)]
       +[\varphi(y+1+u)-\varphi(y)].
\]
Since $y\sim\kappa N$, the same fixed-displacement expansion gives
\begin{align}
 C_N&=-\nu(\log_2N+\log_2e)
       +(1+\nu)(\log_2(\kappa N)+\log_2e)+o(1)\notag\\
    &=\log_2N+\log_2e+(1+\nu)\log_2\kappa+o(1),\\
 S_N&=\log_2N+\log_2(\kappa-1)+\log_2e+g(\nu)+o(1).
\end{align}
Subtracting gives Eq.~\eqref{eq:amp-plob}.

\smallskip\noindent\emph{Additive noise.}
For $\tau=1$ and $v=\xi$, neither coefficient in
Eqs.~\eqref{eq:u-quadratic}--\eqref{eq:w-quadratic} grows linearly
with $N$. Instead,
$u(u+1+\xi)=\xi N$ and $w=u+\xi$, so
\begin{equation}
 u=\sqrt{\xi N}-\frac{1+\xi}{2}+O(N^{-1/2}),~~
 w=u+\xi,~~ t=N-u.
 \label{eq:add-root-limit}
\end{equation}
Both symplectic photon numbers diverge as $\sqrt N$. Accordingly,
\[
 S_N=\log_2(uw)+2\log_2e+o(1)
    =\log_2N+\log_2\xi+2\log_2e+o(1).
\]

The cross entropy now contains growing displacements:
\[
 C_N=\varphi(N-u)+\varphi(N+u+1+\xi)
                   -\varphi(N)-\varphi(N+\xi).
\]
Expand all terms around $N$. The net linear displacement is
$-u+(u+1+\xi)-\xi=1$. The quadratic coefficient is
$[u^2+(u+1+\xi)^2-\xi^2]/2$.
Since $\varphi'(N)=\log_2N+\log_2e$ and
$\varphi''(N)=1/(N\ln2)$, this gives
\begin{equation}
 C_N=\log_2N+\log_2e+
 \frac{u^2+(u+1+\xi)^2-\xi^2}{2N\ln2}+o(1)
 =\log_2N+(1+\xi)\log_2e+o(1).
 \label{eq:additive-cross-entropy-limit}
\end{equation}
The remainders vanish because
$\varphi'''(x)=-1/(x^2\ln2)$ and the displacements are
$O(\sqrt N)$, so their cubic remainders are $O(N^{-1/2})$.
The quadratic term must be retained: it contributes
$\xi\log_2e$ because $u^2/N\to\xi$.
Subtracting $S_N$ yields Eq.~\eqref{eq:add-plob}.

These three calculations keep the channel fixed while its input
energy grows. Together with monotonicity from
Lemma~\ref{lem:concavity}, they prove $\cF(N)\le\Phi$ and the
large-energy limits stated in Sec.~\ref{sec:three}.
The vanishing-excess-noise boundaries are consistent with the
quantum-limited capacity formulas of Ref.~\cite{LossAmplifier};
their separate capacity proofs are not needed here.

\subsection{Hashing and the residual gap}\label{app:hashing}
The benchmark considered here is obtained in two stages: distribute
independent copies of a bipartite state, and then distill them using
one-way classical communication. For a state on $RB$, the two
hashing expressions are $S(B)-S(RB)$ and $S(R)-S(RB)$, with the
communication direction reversed between them
\cite{BDSW,DevetakWinter,RCI}. Their positive maximum is an
achievable rate, rather than an optimization over all two-way
distillation protocols.

Gaussian preparation of the shared state does not constrain its
subsequent distillation to Gaussian operations. This distinction
is required by the impossibility of distilling Gaussian
entanglement using Gaussian operations alone
\cite{EisertNoGo,GiedkeNoGo}. The hashing lower bound is used with
the unrestricted LOCC model specified in the main text.

Sending one arm of a two-mode squeezed vacuum through the channel
gives the covariance matrix
\begin{equation}
 V_{RB}=
 \begin{pmatrix}
 (2N+1)I_2&2\sqrt{\tau N(N+1)}Z\\
 2\sqrt{\tau N(N+1)}Z&(2y+1)I_2
 \end{pmatrix},
 ~~ Z=\operatorname{diag}(1,-1).
 \label{eq:gaussian-covariance}
\end{equation}
Its symplectic eigenvalues, in either order, are
\[
 \sqrt{L^2-4\tau N(N+1)}+(N-y)=2u+1,~~
 \sqrt{L^2-4\tau N(N+1)}-(N-y)=2w+1.
\]
Thus $S(R)=g(N)$, $S(B)=g(y)$, and
$S(RB)=g(u)+g(w)$.
The two directions of one-way hashing give
$S(B)-S(RB)$ and $S(R)-S(RB)$, respectively.
Choosing the larger and allowing the trivial zero rate gives
Eq.~\eqref{eq:hashing}
\cite{DevetakWinter,PGPBL,RCI}.

To explain the infinite-dimensional use of hashing, start with
finite optical cutoffs. Truncate the squeezed input in its
Schmidt number and normalize it. This is a local preparation by
Alice and decreases the transmitted mean. After the channel,
apply a local vacuum-replacement cutoff to Bob's output.
The resulting reference and output spaces are finite-dimensional,
so the finite-dimensional hashing theorem applies directly.
The truncation mechanisms are those used in
Ref.~\cite[Appendices B 4 and E 4]{LossAmplifier}.

Removing Bob's cutoff at fixed input cutoff, and then removing
the input cutoff, gives trace-norm convergence to $\omega_N$.
The reference and output mean photon numbers also converge:
the input mean increases to $N$, the untruncated output mean is
its affine image $\tau N+v$, and vacuum replacement removes only
the finite-energy tail at each fixed input cutoff.
Their sum gives convergence of the joint mean.

These mean bounds justify entropy convergence, which does not
follow from trace-norm convergence alone. For any of the relevant
subsystems, choose a fixed full-rank thermal reference state
$\theta_\beta=Z_\beta^{-1}\exp(-\beta\hat H)$, where $\hat H$ is
the total photon number on that subsystem.
Then
\[
 D(\rho\Vert\theta_\beta)
 =-S(\rho)+\frac{\beta}{\ln2}\Tr(\hat H\rho)+\log_2Z_\beta.
\]
Lower semicontinuity of relative entropy and convergence of the
mean energy give the upper entropy limit; lower semicontinuity
of entropy gives the lower limit. For the joint system one may
equivalently use a fixed full-rank product thermal state.
All limiting entropies are finite. The two hashing expressions
therefore converge to the Gaussian expressions above.

Energy-constrained entropy continuity gives a broader formulation
of this type of limiting argument~\cite{WinterContinuity}.
Here it is needed only for the optical reference and output of
the chosen achievable protocol, which have controlled photon
means. It imposes no energy or entropy condition on the memories
of the arbitrary adaptive protocols covered by the converse.

For any strictly smaller positive target rate, first fix
cutoffs large enough that the finite-dimensional hashing rate
exceeds it, and then take the coding blocklength to infinity.
The optical input mean is already within the original budget
at those fixed cutoffs. This order of limits proves achievability
of the benchmark in the stated finite-round operational model.

The difference between the upper bound and the \emph{unclipped}
hashing expression
$I_{\rm hash}=\max\{g(N),g(y)\}-g(u)-g(w)$ has a simple integral form:
\begin{equation}
 \cF_{\tau,v}(N)-I_{\rm hash}(N)
 =\int_t^{\min(N,y)}
      \log_2\frac{L-z}{z}\,\dd z>0,
 \label{eq:hashing-gap}
\end{equation}
for $N>0$ in the strictly noisy non-entanglement-breaking region.
To derive it, define
\[
 B(z)=\varphi(z)+\varphi(L-z)-\varphi(N)-\varphi(y).
\]
Then $B(t)=C_N$ and
$B'(z)=\log_2[z/(L-z)]$.
If $N\le y$, the upper endpoint is $N$ and $B(N)=g(y)$;
if $y\le N$, it is $y$ and $B(y)=g(N)$.
Thus $B(\min(N,y))=\max\{g(N),g(y)\}$, and integration proves
the identity.

For strict positivity, $h(L-h)$ is increasing on
$[0,\min(N,y)]$, and its value minus $\tau N(N+1)$ at the two
candidate endpoints is, respectively,
$N(1+v-\tau)>0$ and $v(N+1)>0$.
Hence $t<\min(N,y)$.
Also $z\le\min(N,y)<L/2$ throughout the integration interval,
so its integrand is positive.
The achievable benchmark is $[I_{\rm hash}]_+$.
When $I_{\rm hash}<0$, the gap to that clipped benchmark is
$\cF(N)$ itself; Eq.~\eqref{eq:hashing-gap} is not a formula
for the clipped gap in that case.

\subsection{Evaluation of the compared upper bounds}\label{app:numerics}
For reproducibility, we give the dilation, squashing operation,
and entropy calculation used for Fig.~\ref{fig:bounds}.
These comparisons use squashed entanglement, defined by
\begin{equation}
 E_{\rm sq}(R:B)_\rho
 =\frac12\inf_{\rho_{RBE}:\,\Tr_E\rho_{RBE}=\rho_{RB}}
     I(R;B\mid E)_\rho,
 \label{eq:squashed-definition}
\end{equation}
where $I(R;B\mid E)=S(RE)+S(BE)-S(E)-S(RBE)$ whenever these
entropies are finite~\cite{ChristandlWinter}. An extension can
be constructed by applying a channel to a purifying environment.
Evaluating a particular extension gives an upper bound on the
infimum. Its use for optical capacity converses is developed in
Refs.~\cite{DSW,Goodenough}. In the calculations below the
extension is generated by balanced beam splitters, and the cited
thermal-input optimization results supply the channel bound.
This distinguishes the evaluated comparison from an optimization
over all possible squashing channels.

The constructions of Refs.~\cite[Sec. VI]{DSW}
and~\cite[Eq. (27)]{Goodenough} use opposite orderings of two
quantum-limited components:
\begin{equation}
 \begin{aligned}
 \cG_{\tau,v}&=\cL_T\circ\cA_{G',0},
 &T&=\tau-v,&G'&=\tau/T,\\
 \cG_{\tau,v}&=\cA_{G,0}\circ\cL_a,
 &G&=1+v,&a&=\tau/G.
 \end{aligned}
 \label{eq:comparison-decompositions}
\end{equation}
The first is the DSW ordering and the second the GEW ordering.
In the strictly noisy non-entanglement-breaking region,
$0<T,a<1$ and $G',G>1$.
For example, the first ordering gives output mean
$T[G'N+(G'-1)]=\tau N+v$, whereas the second gives
$GaN+(G-1)=\tau N+v$.
Their covariance actions also coincide with
Eq.~\eqref{eq:channel}. These components are auxiliary
dilations of a noisy channel, not additional capacity cases
studied in this paper.

\smallskip\noindent\emph{1. The symmetric squashing extension.}
Implement either ordering with two vacuum environment inputs.
Split each environment output with another vacuum on a balanced
beam splitter. Denote the retained environment outputs by
$E'_1,E'_2$ and the other halves by $F'_1,F'_2$.
If $R$ purifies the thermal signal input, the full state
$RBE'_1E'_2F'_1F'_2$ is pure.
Conditional-entropy duality gives
\[
 I(R;B\mid E'_1E'_2)
 =S(B\mid E'_1E'_2)+S(B\mid F'_1F'_2).
\]
Balanced splitting makes the two conditional entropies equal
up to local phase rotations. Evaluating one half of this conditional
mutual information therefore gives
\begin{equation}
 U(N)=S(BE'_1E'_2)-S(E'_1E'_2).
 \label{eq:squash-evaluation}
\end{equation}
We use the thermal-input optimization and the
energy-constrained converse for these phase-insensitive Gaussian
squashing constructions from Refs.~\cite{DSW,Goodenough}.
The two orderings give $U_{\rm DSW}$ and $U_{\rm GEW}$.
Only the displayed squashing maps are evaluated.

\smallskip\noindent\emph{2. Construct the covariance matrix.}
An explicit reference mode $R$ is unnecessary when computing
the two entropies in Eq.~\eqref{eq:squash-evaluation}: both
depend only on the thermal signal marginal and environment modes.
Start from the five-mode covariance
\[
 V_0=(2N+1)I_2\oplus I_2\oplus I_2\oplus I_2\oplus I_2.
\]
Mode $0$ is the signal, modes $1,2$ are dilation vacua, and
modes $3,4$ are the squashing vacua. On the indicated pair of
modes embed
\begin{equation}
 S_{\rm BS}(q)=
 \begin{pmatrix}
 \sqrt qI_2&\sqrt{1-q}I_2\\
 -\sqrt{1-q}I_2&\sqrt qI_2
 \end{pmatrix},
 ~~
 S_{\rm amp}(G)=
 \begin{pmatrix}
 \sqrt GI_2&\sqrt{G-1}Z\\
 \sqrt{G-1}Z&\sqrt GI_2
 \end{pmatrix},
 \label{eq:gaussian-matrices}
\end{equation}
with identity on the other modes.
For DSW, apply $S_{\rm amp}(G')$ on $(0,1)$, followed by
$S_{\rm BS}(T)$ on $(0,2)$.
For GEW, apply $S_{\rm BS}(a)$ on $(0,1)$, followed by
$S_{\rm amp}(G)$ on $(0,2)$.
In both cases finish with $S_{\rm BS}(1/2)$ on $(1,3)$ and
$(2,4)$. At each step update $V\mapsto SVS^T$.
After these transformations, modes $(0,1,2)$ are
$(B,E'_1,E'_2)$.

\smallskip\noindent\emph{3. Extract entropies and form the comparison.}
Take the principal covariance submatrix on modes $(0,1,2)$
and, separately, on modes $(1,2)$.
For an $\ell$-mode covariance $V_{\rm red}$, let its symplectic
eigenvalues be $\lambda_1,\ldots,\lambda_\ell\ge1$.
They are the positive magnitudes from the paired spectrum of
$i\Omega_\ell V_{\rm red}$, with
$\Omega_\ell=\bigoplus_{j=1}^\ell
\begin{psmallmatrix}0&1\\-1&0\end{psmallmatrix}$,
and its entropy is
\[
 S(V_{\rm red})=\sum_{j=1}^\ell
          g\!\left(\frac{\lambda_j-1}{2}\right).
\]
The difference of the two entropies gives the value in
Eq.~\eqref{eq:squash-evaluation}. Taking the minimum of the two
orderings and the corresponding PLOB value then gives
$U_{\rm prev}$ in Eq.~\eqref{eq:previous}.
In the entanglement-breaking region the capacity bound is set
to zero directly; the DSW parametrization is not evaluated at
the singular point $T=0$.

For numerical evaluation of the sector expression, the smaller
root can be written without subtracting nearly equal roots:
\[
 t=\frac{2\tau N(N+1)}
        {L+\sqrt{L^2-4\tau N(N+1)}}.
\]
At very large energies, the entropy differences themselves may
still require increased arithmetic precision. The analytic
limits in Appendix~\ref{app:highenergy} give independent checks
on that regime.

\end{document}